\documentclass[11pt,a4paper]{article}
\usepackage{amsfonts,amsmath,amssymb,amsthm}
\usepackage{latexsym,amscd}
\usepackage{amsbsy}
\usepackage{CJK}
\usepackage{indentfirst,mathrsfs}
\usepackage{latexsym,amscd}
\usepackage{amsbsy}
\usepackage[noadjust]{cite}
\usepackage{color}
\usepackage{multirow}
\usepackage{makecell}
\usepackage{float}
\usepackage{booktabs}
\usepackage{graphicx}
\usepackage{rotating}
\usepackage[margin=2.5cm]{geometry}

\newcommand\reallywidehat[1]{
\savestack{\tmpbox}{\stretchto{
\scaleto{
\scalerel*[\widthof{\ensuremath{#1}}]{\kern-.6pt\bigwedge\kern-.6pt}
{\rule[-\textheight/2]{1ex}{\textheight}}
}{\textheight}
}{0.5ex}}
\stackon[1pt]{#1}{\tmpbox}
}

\numberwithin{equation}{section}
\newtheorem{theo}{Theorem}[section]
\newtheorem{lem}[theo]{Lemma}
\newtheorem{coro}[theo]{Corollary}

\newtheorem{definition}[theo]{Definition}
\newtheorem{example}[theo]{Example}
\newtheorem{rem}[theo]{Remark}

\title{New constructions of cyclic constant-dimension subspace codes based on Sidon spaces and subspace polynomials}
\author{Gang Wang\textsuperscript{$^*$} \and Ming Xu \and You Gao}
\date{\small College of Science, Civil Aviation University of China, 300300, Tianjin, China. \\ E-mail: gwang06080923@mail.nankai.edu.cn; xuhaoshu2008@163.com; gao$_{-}$you@263.net.\\
$^*$Corresponding author}

\begin{document}

\maketitle

\begin{abstract}
Subspaces coding plays an important role in error correction of random network coding. To study their properties and find good constructions, the notion of cyclic subspace codes was introduced using the extension field structure of the ambient space. Those cyclic constant-dimension subspace codes (CDCs) with optimal minimum distances may have additional structures that could be effectively applied in encoding and decoding algorithms (see \cite{Trautmann2013,Etzion2011,Braun2016,Kohnert2008}). 
In this paper, two new constructions of Sidon spaces are given by tactfully adding new parameters and flexibly varying the number of parameters. Under the parameters $ n= (2r+1)k, r \ge2 $ and $p_0=\max \{i\in \mathbb{N}^+: \lfloor \frac{r}{i}\rfloor>\lfloor \frac{r}{i+1} \rfloor  \}$, the first construction produces a cyclic CDC in $\mathcal{G}_q(n, k)$ with minimum distance $2k-2$ and size $\frac{\left((r+\sum\limits_{i=2}^{p_0}(\lfloor \frac{r}{i}\rfloor-\lfloor \frac{r}{i+1} \rfloor))(q^k-1)(q-1)+r\right)(q^k-1)^{r-1}(q^n-1)}{q-1}$.
Given parameters $n=2rk,r\ge 2$ and if $r=2$, $p_0=1$, otherwise, $p_0=\max\{ i\in \mathbb{N}^+: \lceil\frac{r}{i}\rceil-1>\lfloor \frac{r}{i+1} \rfloor  \}$, a cyclic CDC in $\mathcal{G}_q(n, k)$ with minimum distance $2k-2$ and size $\frac{\left((r-1+\sum\limits_{i=2}^{p_0}(\lceil \frac{r}{i}\rceil-\lfloor \frac{r}{i+1} \rfloor-1))(q^k-1)(q-1)+r-1\right)(q^k-1)^{r-2}\lfloor \frac{q^k-2}{2}\rfloor(q^n-1)}{q-1}$ is produced by the second construction.
The sizes of our cyclic CDCs are larger than the best known results. In particular, in the case of $n=4k$, when $k$ goes to infinity, the ratio between the size of our cyclic CDC and the Sphere-packing bound (Johnson bound) is approximately equal to $\frac{1}{2}$. 
Moreover, for a prime power $q$ and positive integers $k,s$ with $1\le s< k-1$, a cyclic CDC in $\mathcal{G}_q(N, k)$ of size $e\frac{q^N-1}{q-1}$ and minimum distance $\ge 2k-2s$ is provided by subspace polynomials, where $N,e$ are positive integers. Our construction generalizes previous results and, under certain parameters, provides cyclic CDCs with larger sizes or more admissible values of $ N $ than constructions based on trinomials.

\end{abstract}



\section{Introduction}

Ahlswede et al. \cite{Ahlswede2000} 
proposed the random network coding, which has been proven effective in a incoherent network. Let $q$ be a prime power and $\mathbb{F}_q$ be the finite field with $q$ elements. Suppose that $n$ is a positive integer and $\mathbb{F}_{q}^n$ is the $n$-dimensional vector space over $\mathbb{F}_q$. Koetter and Kschischang \cite{Koetter2008} 
first introduced the concept of subspace codes, which treat messages as subspaces of some fixed vector space $\mathbb{F}_{q}^n$. Hence, a code is a set of subspaces of $\mathbb{F}_{q}^n$. In addition, the authors provided an algebraic approach to random network coding in non-coherent networks and gave error correction and the corresponding transmission model. Since then, some researchers focused on the constructions and bounds of subspace codes (cf. \cite{Ben-Sasson2016,Ben-Sasson2010,Gluesing-Luerssen2021}). 
Let $\mathcal{G}_q(n,k)$ be the set of all $k$-dimensional subspaces of $\mathbb{F}_{q}^n$. For two subspaces $U$ and $V$ of $\mathcal{G}_q(n,k)$, $U \cap V$ denotes the intersection subspace of $U$ and $V$. A nonempty subset $\mathcal{C}$ of $\mathcal{G}_q(n,k)$ is a CDC under the subspace metric $d(U,V)=\mathrm{dim}U+\mathrm{dim}V-2\mathrm{dim}(U \cap V)$, whose minimum (subspace) distance is defined as $ d(\mathcal{C})=\mathrm{min}\{d(U,V):U,V\in \mathcal{C},U\ne V \}.$
For convenience, a CDC of $\mathcal{G}_q(n,k)$ with minimum (subspace) distance $d$ is denoted by $(n,d,k)_q$-CDC.
Moreover, noting that the extension field $\mathbb{F}_{q^n}$ is also a $\mathbb{F}_{q}$-vector space with dimension $n$ and is algebraically richer than $\mathbb{F}_{q}^n$, researchers are keen on studying subspace codes composed of subspaces of $\mathbb{F}_{q^n}$.
Obviously, $\mathbb{F}_{q^n}$ is equivalent to $\mathbb{F}_{q}^n$ as a vector space over $\mathbb{F}_q$.
For a subspace $U \in \mathcal{G}_q(n,k)$ and a nonzero element $\alpha \in \mathbb{F}_{q^n}^{*}\triangleq \mathbb{F}_{q^n}\setminus \{0 \}$, the cyclic shift of $U$ by $\alpha $ is $\alpha U\triangleq \{\alpha u : u\in U \}$, which is clearly a subspace of the same dimension as $U$. Two cyclic shifts are called distinct if they form two different subspaces. A subspace code $\mathcal{C}$ is said to be cyclic if $\alpha U \in \mathcal{C}$ for any $\alpha \in \mathbb{F}_{q^n}^{*}$ and any $U \in \mathcal{C}$. The orbit of $U$ is orb$(U)\triangleq \{ \beta U : \beta \in \mathbb{F}_{q^n}^{*}\}$ and its cardinality is $\frac{q^n-1}{q^k-1}$ for some integer $k$ which divides $n$. It is known that a cyclic single orbit code orb($U$) has minimum distance $2k$ if and only if $U$ is a cyclic shift of $\mathbb{F}_{q^k}$ and $k$ must be a divisor of $n$ \cite{Otal2017}. 
 In other cases, the minimum distance of orb($U$) is less than or equal to $2k-2$ and the size of orb$(U)$ is less than or equal to $\frac{q^n-1}{q-1}$.
In particular, it was previously suggested  
that cyclic $(n,d,k)_q$-CDCs may present a useful structure that can be applied efficiently for the purpose of coding. A cyclic subspace code is a union of cyclic orbit codes \cite{Etzion2011,Braun2016,Kohnert2008}. 
In the paper \cite{Trautmann2013}, 
the authors presented two decoding algorithms, named rank-based algorithm and syndrome decoding. They explained how to decode irreducible cyclic orbit codes and determined the complexities of the proposed algorithms.
Therefore, many researchers focus on constructing cyclic $(n,2k-2,k)_q$-CDCs by unionising single orbit codes whose minimum distances can reach $2k-2$ as many as possible.

Ben-Sasson et al. \cite{Ben-Sasson2016} 
constructed a cyclic $(n,2k-2,k)_q$-CDC with size $\frac{q^n-1}{q-1}$ by using some subspace polynomials. Otal and $\ddot{\mathrm{O}}$zbudak \cite{Otal2017} 
provided a cyclic $(n,2k-2,k)_q$-CDC through the union of distinct cyclic single orbit codes generated by $r$ distinct subspace polynomials. Therefore, the size of this cyclic $(n,2k-2,k)_q$-CDC is $r \frac{q^n-1}{q-1}$. Chen and Liu \cite{Chen2018} 
used more general subspace polynomials to produce several cyclic $(n,2k-2,k)_q$-CDCs. Zhao and Tang \cite{Zhao2019} 
explored the parameters of a class of the cyclic CDCs whose subspace polynomials are more generalized.

Roth et al. \cite{Roth2017} 
provided new methods to construct cyclic $(n,d,k)_q$-CDCs using Sidon spaces. They resolved a few cases of Conjecture 2.4 proposed in \cite{Trautmann2013} 
for $ n > 2k$. By combining several cyclic $(n,2k-2,k)_q$-CDCs, Zhang and Cao \cite{Zhang2021} 
established a large cyclic $(n,2k-2,k)_q$-CDC with cardinality $2\lfloor\frac{q-1}{2} \rfloor\frac{q^n-1}{q-1}$. Feng and Wang \cite{Feng2021} 
gave a cyclic $(n,2k-2,k)_q$-CDC with size $(\lfloor\frac{n}{2k} \rfloor-1) \frac{q^k (q^n-1)}{q-1}$ for $n \ge 3k$ through the union of optimal cyclic orbit codes generated by Sidon spaces. Li and Liu \cite{Li2023} 
presented a criterion that can be applied to determine whether the sum of some distinct Sidon spaces is again a Sidon space. Based on this result, they obtained cyclic $(n,2k-2,k)_q$-CDCs via the sum of several Sidon spaces. For more information on constructions and bounds for subspace codes, the interested readers are referred to \cite{Cai2002,Gluesing-Luerssen2015, Ho2006,Koetter2002,Lidl1997,Niu2020,Zhang2022,Zhang2022b,Zhang2023,Zhang2023b,Zullo2023,Lao2022,Niu2022,Zhang2023c,Liu2023}.

The paper aims to generalize the previous works in \cite{Li2024,Yu2024} 
and obtain new cyclic $(n,2k-2,k)_q$-CDCs with larger sizes. We give new constructions of Sidon spaces and, based on these, construct several new cyclic $(n,2k-2,k)_q$-CDCs by taking the union of orbits of Sidon spaces. More explicitly, we provide a cyclic $(n,2k-2,k)_q$-CDC (see Theorem 3.3) of $\mathcal{G}_q(n,k)$ with size
\begin{flalign*}
\frac{\left((r+\sum\limits_{i=2}^{p_0}(\lfloor \frac{r}{i}\rfloor-\lfloor \frac{r}{i+1} \rfloor))(q^k-1)(q-1)+r\right)(q^k-1)^{r-1}(q^n-1)}{q-1},
\end{flalign*}
as the case $n=(2r+1)k$ and $p_0=\max \{i\in \mathbb{N}^+: \lfloor \frac{r}{i}\rfloor>\lfloor \frac{r}{i+1} \rfloor  \}$.
Moreover, we also present a cyclic $(n,2k-2,k)_q$-CDC (see Theorem 3.7) of $\mathcal{G}_q(n,k)$ with size
\begin{flalign*}
\frac{\left((r-1+\sum\limits_{i=2}^{p_0}(\lceil \frac{r}{i}\rceil-\lfloor \frac{r}{i+1} \rfloor-1))(q^k-1)(q-1)+r-1\right)(q^k-1)^{r-2}\lfloor \frac{q^k-2}{2}\rfloor(q^n-1)}{q-1},
\end{flalign*}
where $n=2rk$ and if $r=2$, $p_0=1$, otherwise, $p_0=\max\{ i\in \mathbb{N}^+: \lceil\frac{r}{i}\rceil-1>\lfloor \frac{r}{i+1} \rfloor  \}$. We emphasize that the sizes of our two cyclic $(n,2k-2,k)_q$-CDCs are larger than the best known results. Previously best known results about the sizes of cyclic $(n,2k-2,k)_q$-CDCs, together with the main results in this paper, are collected in Table 1.
Moreover, for a prime power $q$ and positive integers $k,s$ with $1\le s< k-1$, a cyclic CDC in $\mathcal{G}_q(N, k)$ of size $e\frac{q^N-1}{q-1}$ and minimum distance $\ge 2k-2s$ is provided by subspace polynomials, where $N,e$ are positive integers. 
Our construction is more general than the result in~\cite{Zhao2019}, and under certain parameter settings, it provides optimal cyclic CDCs that either support a wider range of values for $ N $ or attain larger size than previous constructions based on trinomials (see Example 3.14).


The rest of this paper is organized as follows. Section 2 presents preliminary results on Sidon spaces, subspace polynomials, and cyclic subspace codes. Sections 3.1 and 3.2 provide two new constructions of Sidon spaces and two large cyclic $(n,2k-2,k)_q$-CDCs. Section 3.3 provides a new construction of subspace polynomials and a large cyclic $(n,d,k)_q$-CDC.
Section 4 concludes this paper.





\begin{sidewaystable}[p]
\renewcommand{\arraystretch}{1.5}
\resizebox{1\textwidth}{1.6in}{\begin{tabular}{|c|c|c|c|}
\hline
\rule{0pt}{14pt} 
Parameters&Our constructions&\ The best known results& Differences \\ \hline
\multirow{3}{*}{$n=(2r+1)k$}&$\frac{\left((r+\sum\limits_{i=2}^{p_0}(\lfloor \frac{r}{i}\rfloor-\lfloor \frac{r}{i+1} \rfloor))(q^k-1)(q-1)+r\right)(q^k-1)^{r-1}(q^n-1)}{q-1}$& \multirow{2}{*}{$r\left [(q^k-1)^r(q^n-1)+\frac{(q^k-1)^{r-1}(q^n-1)}{q-1}  \right ] $} &\multirow{2}{*}{$\small \left(\sum\limits_{i=2}^{p_0}(\lfloor \frac{r}{i}\rfloor-\lfloor \frac{r}{i+1} \rfloor)\right)(q^k-1)^{r}(q^n-1)$}\\
~&{\tiny($p_0=\max \{i\in \mathbb{N}^+: \lfloor \frac{r}{i}\rfloor>\lfloor \frac{r}{i+1} \rfloor  \}$)} &~&~\\
\rule{0pt}{14pt}
~&\footnotesize (Theorem 3.3)&\footnotesize \cite[Theorem 2.2]{Yu2024}&~\\ \hline
\multirow{3}{*}{$n=2rk$}&{ \scriptsize$\frac{\left((r-1+\sum\limits_{i=2}^{p_0}(\lceil \frac{r}{i}\rceil-\lfloor \frac{r}{i+1} \rfloor-1))(q^k-1)(q-1)+r-1\right)(q^k-1)^{r-2}\lfloor \frac{q^k-2}{2}\rfloor(q^n-1)}{q-1}$}&\multirow{2}{*}{$\lfloor \frac{q^k-2}{2}\rfloor(r-1)(q^k-1)^{r-1}(q^n-1) $ }& \multirow{3}{*}{$\frac{\left((\sum\limits_{i=2}^{p_0}(\lceil \frac{r}{i}\rceil-\lfloor \frac{r}{i+1} \rfloor-1))(q^k-1)(q-1)+r-1\right)(q^k-1)^{r-2}\lfloor \frac{q^k-2}{2}\rfloor(q^n-1)}{q-1}$}\\
~&{\tiny(if $r=2$, $p_0=1$, otherwise, $p_0=\max\{ i\in \mathbb{N}^+: \lceil\frac{r}{i}\rceil-1>\lfloor \frac{r}{i+1} \rfloor  \}$)} & ~&~\\
~&\footnotesize (Theorem 3.7)&\footnotesize \cite[Theorem 3.3]{Yu2024}&~\\ \hline
\multirow{2}{*}{$n=5k$}&$3(q^k-1)^2(q^n-1)+\frac{2(q^k-1)(q^n-1)}{q-1}$&$(q^k-1)(3q^k-2)\frac{q^n-1}{q-1} $&\multirow{2}{*}{$((q^k-1)(3q-6)+1)\frac{(q^k-1)(q^n-1)}{q-1}$}\\
~&\footnotesize (Remark 3.4)&\footnotesize \cite[Theorem 3.7]{Li2024}&~\\ \hline
\end{tabular}}

\caption{The sizes of $(n,2k-2,k)_q$-CDCs}
\end{sidewaystable}
\section{Preliminaries}

In this section, we present some notions and lemmas that will be used in the following section. For a set $S$, the cardinality of $S$ is denoted by $|S|$. Given a real number $r$, let $\lceil r \rceil$ denote the least integer greater than or equal to $r$, and $\lfloor r \rfloor$ denote the greatest integer less than or equal to $r$.

The following lemma can determine the cardinality of a cyclic simple orbit subspace code.
\begin{lem}\cite{Otal2017} 
Let $U\in \mathcal{G}_q(n,k)$. Then $\mathbb{F}_{q^d}$ is the largest field such that $U$ is also $\mathbb{F}_{q^d}$-linear if and only if $|orb(U)|=\frac{q^n-1}{q^d-1}$.
\end{lem}
Notice that the minimum distance of orb$(U)$ equals to $2k-2s$, where $0 \le s \le k$ is the maximum value of the set $\{$dim$(U \cap \alpha U ) : \forall \alpha \in \mathbb{F}_{q^n}^{*} \}$. If $d=s=1$, then we call orb($U$) the optimal cyclic orbit code.

\begin{definition}\cite{Bachoc2017} 
A subspace $U \in \mathcal{G}_q(n,k)$ is called a Sidon space if for any nonzero elements $a, b, c, d \in U$, if $ab = cd$, then $\left \{ a\mathbb{F}_q, b\mathbb{F}_q \right \} =
\left \{c\mathbb{F}_q, d\mathbb{F}_q \right \} $.
\end{definition}
The choice to name these subspaces after Simon Sidon (1892-1941) draws inspiration from the closely related Sidon sets. From Definition 2.2, it is easy to see that a Sidon space is a subspace $U \in \mathcal{G}_q(n,k)$ such that the product of any two nonzero elements of $U$ has a unique factorization over $U$ up to a constant multiplier from $\mathbb{F}_q$. Furthermore, Roth et al. in \cite{Roth2017} 
 proposed to construct a optimal cyclic orbit code by Sidon spaces.

\begin{lem}\cite{Roth2017} 
For a subspace $U \in \mathcal{G}_q(n,k)$, the code orb$(U)$ is of size $\frac{q^n-1}{q-1}$ and minimum distance $2k-2$ if and only if $U$ is a Sidon space.
\end{lem}
Lemma 2.3 shows that a cyclic $(n,2k-2,k)_q$-CDC with cardinality $\frac{q^n-1}{q-1}$ is equivalent to the orbit of a Sidon space. Moreover, Roth et al. in \cite{Roth2017} presented a criterion that can be applied to determine whether the union of cyclic $(n,2k-2,k)_q$-CDCs generated by the orbits of Sidon spaces changes the minimum distance.

\begin{lem}\cite{Roth2017} 
The following two conditions are equivalent for any distinct subspaces U and V in $\mathcal{G}_q(n,k)$.

(1) $dim(U \cap \alpha V)\le 1, ~for ~any~ \alpha \in \mathbb{F}_{q^n}^{*}$.

(2) For any nonzero $a,c \in U$ and nonzero $b,d\in V$, the equality $ab=cd$ implies that $a\mathbb{F}_q = c\mathbb{F}_q$ and $b\mathbb{F}_q = d\mathbb{F}_q$.
\end{lem}
Based on Lemma 2.4, we can construct a large cyclic $(n,2k-2,k)_q$-CDC by merging distinct cyclic $(n,2k-2,k)_q$-CDCs as many as possible.
\begin{lem}\cite{Yu2024} 
Let $q$ be a prime power, $k$ a positive integer, $c$ a nonzero element of $\mathbb{F}_{q^k}$,
and $\xi $ a primitive element of $\mathbb{F}_{q^k}$. Then there is a subset $ A \subset \left \{ 0, 1, . . . , q^k-2 \right \} $ of size $\left \lfloor \frac{q^k-2}{2} \right \rfloor $ satisfying~$ c\xi ^{i+j} \ne 1$ for any $i,j\in A$.
\end{lem}
Lemma 2.5 helps us calculate the cardinality of the codes presented in Section 3.2. Furthermore, we recall the upper bounds of CDCs.

\begin{lem}\cite{Etzion2011} (Sphere-packing bound) 
Let $\mathcal{C} \subseteq \mathcal{G}_q(n,k)$ be a CDC with minimum distance $2\delta +2$. Then
$$|\mathcal{C}| \le \frac{\genfrac[]{0pt}{0}{n}{k-\delta}_q }{\genfrac[]{0pt}{0}{k}{k-\delta}_q },$$ where
$\genfrac[]{0pt}{0}{n}{k}_q \triangleq \prod \limits_{i=0}^{k-1} \frac{q^{n-i}-1}{q^{k-i}-1}.$

(Johnson bound)
Let $\mathcal{C} \subseteq \mathcal{G}_q(n,k)$ be a CDC with minimum distance $2\delta$. Then
$$|\mathcal{C}| \le \prod_{i=0}^{k-\delta}\frac{q^{n-i}-1}{q^{k-i}-1}.$$
\end{lem}

Furthermore, there is another tool called subspace polynomials, also known as linearized polynomials, which can be used to construct cyclic CDCs. A polynomial of the form $f(x)=\sum \limits_{i=0}^k a_i x^{q^i}\in \mathbb{F}_{q^n}[x]$ is called a linearized polynomial (or $q$-polynomial) and $k$ is called the $q$-degree of $f$ if $a_k\ne 0$. 
The roots of a linearized polynomial $f(x)\in \mathbb{F}_{q^n}[x]$ form a subspace of an extension field $\mathbb{F}_{q^N}$ of $\mathbb{F}_{q^n}$. 
Each root of $f(x)$ has the same multiplicity, which is either $1$ or a power of $q$. 
Conversely, each subspace of $\mathbb{F}_{q^N}$ corresponds to a linearized  polynomial over $\mathbb{F}_{q^N}$. 
Hence, researchers are particularly interested in linearized polynomials that have simple roots with respect to some field $\mathbb{F}_{q^n}$.

Beyond their role in subspace codes, linearized polynomials are also deeply connected to rank-metric codes. Each linearized polynomial over a finite field defines an \( \mathbb{F}_q \)-linear transformation, which naturally corresponds to a matrix under the rank metric. As a result, rank-metric codes can be viewed as subsets of the algebra of linearized polynomials, where the rank of each codeword coincides with the rank of the associated linear map. Notably, classical maximum rank distance (MRD) codes such as Gabidulin and twisted Gabidulin codes arise from carefully structured families of linearized polynomials with selected supports and coefficients.

\begin{definition}\cite{Ben-Sasson2016} 
A monic linearized polynomial $P$ with coefficients in $\mathbb{F}_{q^n}$ is called a subspace polynomial with respect to $\mathbb{F}_{q^n}$ if the following equivalent conditions hold:

(1) $P$ divides $x^{q^n}-x$.

(2) $P$ splits completely over $\mathbb{F}_{q^n}$ and all its roots have multiplicity 1.
\end{definition}
Given a subspace polynomial with coefficients in \( \mathbb{F}_{q^n} \), let \( \mathbb{F}_{q^N} \) be the field containing all its roots, i.e., \( \mathbb{F}_{q^N} \) is an extension field of \( \mathbb{F}_{q^n} \). 
Moreover, each subspace in $\mathcal{G}_q(N, k)$ corresponds uniquely to a subspace polynomial of $q$-degree k, which leads to the following lemma.
\begin{lem}\cite{Ben-Sasson2016} 
Two subspaces are equal if and only if their corresponding subspace polynomials are equal.
\end{lem}
Lemma 2.8 shows that it can be determined whether two subspaces are equal by comparing their corresponding subspace polynomials. It is clear that $P_{\alpha V}(x)=\prod\limits_{v\in V}(x-\alpha v)$ for $\alpha \in \mathbb{F}_{q^N}^*$.
\begin{lem}\cite{Ben-Sasson2016} 
If $V \in \mathcal{G}_q(N,k)$ and $\alpha \in \mathbb{F}_{q^N}^*$ then $P_{\alpha V}(x)=\prod\limits_{v\in V}(x-\alpha v)=\alpha^{q^k}\cdot P_V(\alpha^{-1}x)$. That is, if $P_V(x)=x^{q^k}+\sum\limits_{j=0}^ia_j x^{q^j}$ then $P_{\alpha V}(x)=x^{q^k}+\sum\limits_{j=0}^i\alpha^{q^k-q^j}a_jx^{q^j}$.
\end{lem}
Lemma 2.9 helps to determine the dimension of $V \cap \alpha V$, by discussing the degree of the greatest common divisor of $P_V(x)$ and $P_{\alpha V}(x)$.

\section{Constructions of new cyclic CDCs}
\subsection{Constructions of cyclic CDCs from Sidon spaces}
In this subsection, we construct a new cyclic $(n,2k-2,k)_q$-CDC with larger size by Sidon spaces.

Firstly, we recall some basic definitions and results. Let $\xi$ be a primitive element of $\mathbb{F}_{q^k}$.
For any positive integer $k \ge2$ and $n = (2r+1)k$ with $r \ge2$, suppose that $l$ and $p$ are positive integers, $p_0=\max \{i\in \mathbb{N}^+: \lfloor \frac{r}{i}\rfloor>\lfloor \frac{r}{i+1} \rfloor  \}$ and $\gamma$ is a root of an irreducible polynomial of degree $2r+1$ over $\mathbb{F}_{q^k}$. Let
\begin{flalign}
&U_{\delta_{1},\delta_2,...,\delta_{r},\theta,l}^{p}=\left\{u+\sum_{a=1}^{p}(\theta u^q+u)\delta_{al}\gamma^{al}+\sum_{b=1,b\notin \{ l,2l,...,pl\}}^{r}u
\delta_{b}\gamma^b: u \in \mathbb{F}_{q^k}\right\},
\end{flalign}
where $\delta_1,\delta_2,...,\delta_r \in \mathbb{F}_{q^k}^{*}, ~1\le a,b\le r,~\theta \in \{1,\xi,...,\xi^{q-2}\}$ and $1\le p \le p_0 $.
Moreover, if $p=1,$ then $1 \le l \le r,$ otherwise, $\frac{r}{p+1} < l \le \frac{r}{p}$.
Let
\begin{flalign}
&V_{\delta_{1},\delta_2,...,\delta_{r},l}~~=\left\{v+v^q \gamma^{l}
+\sum_{b=1,b\ne l}^{r}v\delta_{b}\gamma^b : v \in \mathbb{F}_{q^k}\right\},
\end{flalign}
where $\delta_1,\delta_2,...,\delta_r \in \mathbb{F}_{q^k}^{*}$ and $1\le b,l\le r$.

Below we will briefly show that $U_{\delta_{1},\delta_2,...,\delta_{r},\theta,l}^{p}$ and $V_{\delta_{1},\delta_2,...,\delta_{r},l}$ of the above construction are the $k$-dimensional subspaces of $\mathbb{F}_{q^n}.$

Noting that $\forall c,d \in \mathbb{F}_{q}^{*}, ~\alpha,\beta \in \mathbb{F}_{q^k}^{*}$, we have
\begin{flalign*}
& c\alpha+d\beta +\sum_{a=1}^{p}(\theta (c\alpha+d\beta )^q+c\alpha+d\beta )\delta_{al}\gamma^{al}+\sum_{b=1,b\notin \{ l,2l,...,pl\}}^{r} (c\alpha+d\beta )
\delta_{b}\gamma^b \\
& =c(\alpha+\sum_{a=1}^{p}(\theta \alpha^q+\alpha)\delta_{al}\gamma^{al}+\sum_{b=1,b\notin \{ l,2l,...,pl\}}^{r} \alpha
\delta_{b}\gamma^b)\\
&~~~+d(\beta+\sum_{a=1}^{p}(\theta \beta^q+\beta)\delta_{al}\gamma^{al}+\sum_{b=1,b\notin \{ l,2l,...,pl\}}^{r}\beta
\delta_{b}\gamma^b).
\end{flalign*}
Then, $U_{\delta_{1},\delta_2,...,\delta_{r},\theta,l}^{p}$ is a subspace of $\mathbb{F}_{q^n}$.

Since $1, \gamma,...,\gamma^{2r}$ are linearly independent over $\mathbb{F}_{q^k}$, for all $u,v \in \mathbb{F}_{q^k}^*$,
\begin{flalign*}
u+\sum_{a=1}^{p}(\theta u^q+u)\delta_{al}\gamma^{al}+\sum_{b=1,b\notin \{ l,2l,...,pl\}}^{r} u\delta_{b}\gamma^b=v+\sum_{a=1}^{p}(\theta v^q+v)\delta_{al}\gamma^{al}+\sum_{b=1,b\notin \{ l,2l,...,pl\}}^{r} v\delta_{b}\gamma^b,
\end{flalign*}
if and only if $u=v$. Therefore, $|U_{\delta_{1},\delta_2,...,\delta_{r},\theta,l}^{p}|=q^k$.

Similarly, $V_{\delta_{1},\delta_2,...,\delta_{r},l}$ is also a $k$-dimensional subspace of $\mathbb{F}_{q^n}$.
\begin{lem}
Let $U_{\delta_{1},\delta_2,...,\delta_{r},\theta,l}^{p}$ and $V_{\delta_{1},\delta_2,...,\delta_{r},l}$ be the subspaces constructed in (3.1) and (3.2), respectively, then they are both Sidon spaces.
\end{lem}
\begin{proof}
Firstly, we prove that $U_{\delta_{1},\delta_2,...,\delta_{r},\theta,l}^{p}$ is a Sidon space. Let $\bar{u}_1,\bar{u}_2,\bar{v}_1,\bar{v}_2 \in U_{\delta_{1},\delta_2,...,\delta_{r},\theta,l}^{p}$ be four nonzero elements and
\begin{flalign*}
\bar{u}_1&=u_1+\sum_{a=1}^{p}(\theta u_1^q+u_1)\delta_{al}\gamma^{al}+\sum_{b=1,b\notin \{ l,2l,...,pl\}}^{r}u_1\delta_{b}\gamma^b,\\
\bar{u}_2&=u_2+\sum_{a=1}^{p}(\theta u_2^q+u_2)\delta_{al}\gamma^{al}+\sum_{b=1,b\notin \{ l,2l,...,pl\}}^{r}u_2\delta_{b}\gamma^b,\\
\bar{v}_1&=v_1+\sum_{a=1}^{p}(\theta v_1^q+v_1)\delta_{al}\gamma^{al}+\sum_{b=1,b\notin \{ l,2l,...,pl\}}^{r}v_1\delta_{b}\gamma^b,\\
\bar{v}_2&=v_2+\sum_{a=1}^{p}(\theta v_2^q+v_2)\delta_{al}\gamma^{al}+\sum_{b=1,b\notin \{ l,2l,...,pl\}}^{r}v_2\delta_{b}\gamma^b
\end{flalign*}
such that $\bar{u}_1\bar{v}_1=\bar{u}_2\bar{v}_2$. Note that $1, \gamma,...,\gamma^{2r}$ are linearly independent over $\mathbb{F}_{q^k}$. Comparing the coefficients of $\gamma^{0}$ and $\gamma^{l}$ on both sides of the equation $\bar{u}_1\bar{v}_1=\bar{u}_2\bar{v}_2$, we obtain that
\begin{flalign*}
\mathrm{the}~ \mathrm{coefficient}~ \mathrm{at}~ \gamma^0&:u_1v_1=u_2v_2,\\
\mathrm{the}~\mathrm{coefficient}~ \mathrm{at} ~\gamma^{l}&:(\theta u_1^q+u_1)v_1\delta_{l}+(\theta v_1^q+v_1)u_1\delta_{l}+\sum_{a=1}^{l-1}u_1v_1\delta_{a}\delta_{l-a}\\
&~~=(\theta u_2^q+u_2)v_2\delta_{l}+(\theta v_2^q+v_2)u_2\delta_{l}+\sum_{a=1}^{l-1}u_2v_2\delta_{a}\delta_{l-a}.
\end{flalign*}
By substituting the first equation into the second equation, we have
\begin{flalign*}
\gamma^{l}&:(\delta_{l}\theta u_1^qv_1+\delta_{l}u_2v_2)+(\delta_{l}\theta u_1v_1^q+\delta_{l}u_2v_2)+\sum_{a=1}^{l-1}u_2v_2\delta_{a}\delta_{l-a}\\
&~~=(\delta_{l}\theta u_2^qv_2+u_2v_2)+(\delta_{l}\theta u_2v_2^q+u_2v_2)+\sum_{a=1}^{l-1}u_2v_2\delta_{a}\delta_{l-a}.
\end{flalign*}
By simplifying it,
we obtain that
\begin{flalign*}
\gamma^0&:u_1v_1=u_2v_2,\\
\gamma^{l}&:u_1^qv_1+u_1v_1^q=u_2^qv_2+u_2v_2^q.
\end{flalign*}
Assume that $\frac{u_1}{u_2}=\frac{v_2}{v_1}=\lambda \ne 0$, then $u_1=\lambda u_2$ and $v_2=\lambda v_1$. Substitute it into the equation of the coefficients of $\gamma^l$, then $\lambda^q u_2^qv_1+\lambda u_2v_1^q=\lambda u_2^qv_1+\lambda^q u_2v_1^q$, thereby $(\lambda^q -\lambda)u_2^qv_1=(\lambda^q -\lambda)u_2v_1^q$. If $\lambda^q -\lambda=0$, $\frac{u_1}{u_2}=\frac{v_2}{v_1}=\lambda \in \mathbb{F}_q^{*}$. If $\lambda^q -\lambda\ne 0$, $(\frac{u_2}{v_1})^q =\frac{u_2}{v_1}\in \mathbb{F}_q^{*}$. Therefore, we have $\left \{ \bar{u}_1\mathbb{F}_q,\bar{v}_1\mathbb{F}_q \right \} =\left \{\bar{u}_2\mathbb{F}_q, \bar{v}_2\mathbb{F}_q \right \} $, which means that the subspace $U_{\delta_{1},\delta_2,...,\delta_{r},\theta,l}^{p}$ is a Sidon space.

The proof of $V_{\delta_{1},\delta_2,...,\delta_{r},l}$ is similar to the above-mentioned. Hence, the subspaces $ U_{\delta_{1},\delta_2,...,\delta_{r},\theta,l}^{p}$ and $ V_{\delta_{1},\delta_2,...,\delta_{r},l}$ are both Sidon spaces.
\end{proof}

\begin{lem}
Let $U_{\delta_{1},\delta_2,...,\delta_{r},\theta,l}^{p}$ be the subspace constructed in (3.1), where $\delta_1,\delta_2,...,\delta_r \in \mathbb{F}_{q^k}^{*},\theta \in \{1,\xi,...,\xi^{q-2}\}$ and $1\le p \le p_0 $. Moreover, if $p=1,$ then $1 \le l \le r,$ otherwise, $\frac{r}{p+1} < l \le \frac{r}{p}$.
Let $V_{\delta_{1},\delta_2,...,\delta_{r},l}$ represent the subspace constructed in (3.2), where $\delta_1,\delta_2,...,\delta_r \in \mathbb{F}_{q^k}^{*}$ and $1\le l\le r$. 
Any two distinct subspaces $X$ and $Y$ in the set of $\{ U_{\delta_{1},\delta_2,...,\delta_{r},\theta,l}^p\} \cup \{ V_{\delta_{1},\delta_2,...,\delta_{r},l}\} $ satisfy the condition $dim(X \cap \alpha Y) \le 1$ for any $\alpha \in \mathbb{F}_{q^n}^{*}$.
\end{lem}
\begin{proof}
To prove that any two distinct subspaces $X$ and $Y$ satisfy the condition dim$(X \cap \alpha Y)\le 1$, for any $\alpha \in \mathbb{F}_{q^n}^{*},$ we distinguish between three cases that $X,Y \in \{ U_{\delta_{1},\delta_2,...,\delta_{r},\theta,l}^p\};$ $X\in \{ U_{\delta_{1},\delta_2,...,\delta_{r},\theta,l}^p\}, Y\in \{ V_{\delta_{1},\delta_2,...,\delta_{r},l}\}$; $X,Y \in \{ V_{\delta_{1},\delta_2,...,\delta_{r},l}\}$.

Case 1: $X,Y \in \{ U_{\delta_{1},\delta_2,...,\delta_{r},\theta,l}^p\}.$

Let $X=U_{\delta_{1},\delta_2,...,\delta_{r},\theta,l}^{p}$, $Y=U_{\delta_{1}',\delta_{2}',...,\delta_{r}',\theta',l'}^{p'}$ and $(\delta_{1},\delta_2,...,\delta_{r},\theta,l,p) \ne (\delta_{1}',\delta_{2}',...,\delta_{r}',\theta',l',p')$.
Suppose that $\bar{u}_1,\bar{u}_2 \in X$ and $\bar{v}_1,\bar{v}_2 \in Y$ are four nonzero elements and
\begin{flalign*}
&\bar{u}_1=u_1+\sum_{a=1}^{p}(\theta u_1^q+u_1)\delta_{al}\gamma^{al}+\sum_{b=1,b\notin \{ l,2l,...,pl\}}^{r}u_1\delta_{b}\gamma^b,\\
&\bar{u}_2=u_2+\sum_{a=1}^{p}(\theta u_2^q+u_2)\delta_{al}\gamma^{al}+\sum_{b=1,b\notin \{ l,2l,...,pl\}}^{r}u_2\delta_{b}\gamma^b,\\
&\bar{v}_1=v_1+\sum_{a=1}^{p'}(\theta' v_1^q+v_1)\delta_{al'}'\gamma^{al'}+\sum_{b=1,b\notin \{ l',2l',...,p'l'\}}^{r}v_1\delta_{b}'\gamma^b,\\
&\bar{v}_2=v_2+\sum_{a=1}^{p'}(\theta' v_2^q+v_2)\delta_{al'}'\gamma^{al'}+\sum_{b=1,b\notin \{ l',2l',...,p'l'\}}^{r}v_2\delta_{b}'\gamma^b
\end{flalign*}
such that $\bar{u}_1\bar{v}_1=\bar{u}_2\bar{v}_2$. By Lemma 2.4, it suffices to show that $\bar{u}_1\mathbb{F}_q=\bar{u}_2\mathbb{F}_q$ and $\bar{v}_1\mathbb{F}_q=\bar{v}_2\mathbb{F}_q$. We distinguish the following six subcases.

Subcase 1.1: $l \ne l'$.

Without loss of generality, assume that $l <l'. $ Since $1, \gamma,...,\gamma^{2r}$ are linearly independent over $\mathbb{F}_{q^k}$, we obtain that
\begin{flalign*}
\gamma^0&:u_1v_1=u_2v_2,\\
\gamma^{l}&:(\theta u_1^q+u_1)v_1\delta_{l}+u_1v_1\delta_{l}' +\sum_{a=1}^{l-1}u_1v_1\delta_{a}\delta_{l-a}'=(\theta u_2^q+u_2)v_2\delta_{l}+u_2v_2\delta_{l}' +\sum_{a=1}^{l-1}u_2v_2\delta_{a}\delta_{l-a}'.
\end{flalign*}
The above equations can be simplified as follows:
\begin{flalign*}
\gamma^0&:u_1v_1=u_2v_2,\\
\gamma^{l}&:u_1^qv_1=u_2^qv_2.
\end{flalign*}
Then $\frac{u_2}{u_1}=\frac{v_1}{v_2}=(\frac{u_2}{u_1})^q\in \mathbb{F}_q^{*}$.

Subcase 1.2: $l=l'$ and $p\ne p'$.

Without loss of generality, assume that $p <p'$.
Since $1, \gamma,...,\gamma^{2r}$ are linearly independent over $\mathbb{F}_{q^k}$ and $pl <p'l'\le r$, we obtain that
\begin{flalign*}
&\gamma^0~~~~~:u_1v_1=u_2v_2,\\
&\gamma^{p'l'+r}:(\theta' v_1^q+v_1)u_1\delta_r \delta_{p'l'}' +u_1v_1\delta_{p'l'}\delta_{r}' +\sum_{a=p'l'+1}^{r-1}u_1v_1\delta_{a}\delta_{p'l'+r-a}'\\
&~~~~~~~~~~=(\theta' v_2^q+v_2)u_2\delta_r \delta_{p'l'}' +u_2v_2\delta_{p'l'}\delta_{r}' +\sum_{a=p'l'+1}^{r-1}u_2v_2\delta_{a}\delta_{p'l'+r-a}'.
\end{flalign*}
The above equations can be simplified as follows:
\begin{flalign*}
&\gamma^0~~~~~:u_1v_1=u_2v_2,\\
&\gamma^{p'l'+r}:u_1v_1^q=u_2v_2^q.
\end{flalign*} Then $(\frac{v_1}{v_2})^q=\frac{u_2}{u_1}=\frac{v_1}{v_2}\in \mathbb{F}_q^{*}$.

Subcase 1.3: $l=l',p= p'$ and $ \theta \ne \theta' $.

When $p>1$, note that $l> \frac{r}{p+1}$ implies $2pl> (p-1)l+r$. Since $1, \gamma,...,\gamma^{2r}$ are linearly independent over $\mathbb{F}_{q^k}$, we obtain that
\begin{flalign*}
&\gamma^0~~:u_1v_1=u_2v_2,\\
&\gamma^{2pl}:(\theta u_1^q+u_1)(\theta' v_1^q+v_1)\delta_{pl}\delta_{pl}' +\sum_{a=\max\{0,2pl-r\},a\ne pl}^{\min\{r,2pl\}}u_1v_1\delta_{a}\delta_{2pl-a}' \\
&~~~~~~~=(\theta u_2^q+u_2)(\theta' v_2^q+v_2)\delta_{pl}\delta_{pl}' +\sum_{a=\max\{0,2pl-r\},a\ne pl}^{\min\{r,2pl\}}u_2v_2\delta_{a}\delta_{2pl-a}',
\end{flalign*}
where $\delta_0=\delta_0'=1$. The above equations can be simplified as follows:
\begin{flalign*}
&\gamma^0~~:u_1v_1=u_2v_2,\\
&\gamma^{2pl}:\theta u_1^qv_1+\theta'u_1v_1^q=\theta u_2^qv_2+\theta'u_2v_2^q.
\end{flalign*}
Assume that $\frac{u_1}{u_2}=\frac{v_2}{v_1}=\lambda \ne 0$, then $\lambda^q \theta u_2^qv_1+\lambda \theta' u_2v_1^q=\lambda \theta u_2^qv_1+\lambda^q \theta' u_2v_1^q$. This means that $(\lambda^q-\lambda)\theta u_2^qv_1=(\lambda^q-\lambda)\theta' u_2v_1^q$. If $\lambda^q-\lambda=0$, $\frac{u_1}{u_2}=\frac{v_2}{v_1}=\lambda \in \mathbb{F}_q^{*}$. If $\lambda^q-\lambda \ne 0$, $\theta u_2^qv_1=\theta'u_2v_1^q$. It follows that $\frac{\theta}{\theta'}=(\frac{v_1}{u_2})^{q-1}\in \langle\xi^{q-1}\rangle$. Since $\theta \ne \theta'$ and $\theta,\theta' \in \{1,\xi,...,\xi^{q-2}\}$, this contradicts the condition that $\frac{\theta}{\theta'}\in \langle\xi^{q-1}\rangle$. Therefore, $\frac{u_1}{u_2}=\frac{v_2}{v_1} \in \mathbb{F}_q^{*}$.

Subcase 1.4: $l=l',p= p',\theta =\theta',\delta_{jl}= \delta'_{jl'}$ $(\forall j\in \{1,2,...,g-1\})$ and $\delta_{gl} \ne \delta'_{gl'} $, where $g$ is a positive integer such that $g\le p$.

Since $1, \gamma,...,\gamma^{2r}$ are linearly independent over $\mathbb{F}_{q^k}$, we obtain that
\begin{flalign*}
\gamma^0~~&:u_1v_1=u_2v_2,\\
\gamma^{gl}~&:(\theta u_1^q+u_1)v_1\delta_{gl}+(\theta' v_1^q+v_1)u_1\delta_{gl}'+\sum_{a=1}^{g-1}(\theta u_1^q+u_1)(\theta' v_1^q+v_1)\delta_{al}\delta_{gl-al}'
\\&~~+\sum_{b=1,b\notin \{l,2l,...,(g-1)l\}}^{gl-1}u_1v_1\delta_{b}\delta_{gl-b}' \\
&~=(\theta u_2^q+u_2)v_2\delta_{gl}+(\theta' v_2^q+v_2)u_2\delta_{gl}'+\sum_{a=1}^{g-1}(\theta u_2^q+u_2)(\theta' v_2^q+v_2)\delta_{al}\delta_{gl-al}'
\\&~~+\sum_{b=1,b\notin \{l,2l,...,(g-1)l\}}^{gl-1}u_2v_2\delta_{b}\delta_{gl-b}',\\
\gamma^{2pl}&:u_1^qv_1+u_1v_1^q=u_2^qv_2+u_2v_2^q.
\end{flalign*}
The above equations can be simplified as follows:
\begin{flalign*}
\gamma^0~~&:u_1v_1=u_2v_2,\\
\gamma^{gl}~&:\delta_{gl}u_1^qv_1+\delta_{gl}'u_1v_1^q=\delta_{gl}u_2^qv_2+\delta_{gl}'u_2v_2^q,\\
\gamma^{2pl}&:u_1^qv_1+u_1v_1^q=u_2^qv_2+u_2v_2^q.
\end{flalign*}
Then $u_1^qv_1-u_2^qv_2=u_2v_2^q-u_1v_1^q$ and $\delta_{gl}(u_1^qv_1-u_2^qv_2)=\delta_{gl}'(u_2v_2^q-u_1v_1^q)$, thereby $\delta_{gl}(u_2v_2^q-u_1v_1^q)=\delta_{gl}'(u_2v_2^q-u_1v_1^q)$. This means that $(\delta_{gl}-\delta_{gl}')u_2v_2^q=(\delta_{gl}-\delta_{gl}')u_1v_1^q$.
Since $\delta_{gl}-\delta_{gl}'\ne 0$, $u_2v_2^q=u_1v_1^q$. It follows that $\frac{v_2}{v_1}=\frac{u_1}{u_2}=(\frac{v_2}{v_1})^q\in \mathbb{F}_q^{*}$.

Subcase 1.5: $ l=l',p= p',\theta =\theta', \delta_{i} = \delta'_{i}$ $(\forall i \in \{1,2,...,tl\}\cup \{(t+1)l,(t+2)l,...,pl\})$ and $\delta_{m} \ne \delta'_{m}$, where $i,t$ and $m$ are nonnegative integers such that $0\le t \le p-1$ and $m \in \{tl+1,tl+2,...,(t+1)l-1\}$.

Since $1, \gamma,...,\gamma^{2r}$ are linearly independent over $\mathbb{F}_{q^k}$, we obtain that
\begin{flalign*}
&\gamma^0~~~~:u_1v_1=u_2v_2,\\
&\gamma^{l+m}:(\theta u_1^q+u_1)v_1\delta_{l}\delta_m'+(\theta' v_1^q+v_1)u_1\delta_m\delta_{l}'+\sum_{b=\max\{0,l+m-r\},b,l+m-b\notin \{l,2l,..,pl\}}^{\min\{l+m,r\}}u_1v_1\delta_{b}\delta_{l+m-b}'\\
&~~~~~~~~~+\sum_{a=2}^{t+1}\left((\theta u_1^q+u_1)v_1\delta_{al}\delta_{l+m-al}'+(\theta' v_1^q+v_1)u_1\delta_{l+m-al}\delta_{al}'\right)\\
&~~~~~~~~=(\theta u_2^q+u_2)v_2\delta_{l}\delta_m'+(\theta' v_2^q+v_2)u_2\delta_m\delta_{l}'+\sum_{b=\max\{0,l+m-r\},b,l+m-b\notin \{l,2l,..,pl\}}^{\min\{l+m,r\}}u_2v_2\delta_{b}\delta_{l+m-b}'\\
&~~~~~~~~~+\sum_{a=2}^{t+1}\left((\theta u_2^q+u_2)v_2\delta_{al}\delta_{l+m-al}'+(\theta' v_2^q+v_2)u_2\delta_{l+m-al}\delta_{al}'\right),\\
&\gamma^{2pl}~~:u_1^qv_1+u_1v_1^q=u_2^qv_2+u_2v_2^q,
\end{flalign*}
where $\delta_0=\delta_0'=1$. Since $\theta =\theta',\delta_{al}=\delta_{al}'$ and $\delta_{l+m-al}=\delta_{l+m-al}', $ for $a\in \{{2,3,...,t+1}\}$, the above equations can be simplified as follows:
\begin{flalign*}
\gamma^0~~~~&:u_1v_1=u_2v_2,\\
\gamma^{l+m}&:\delta_{m}u_1^qv_1+\delta_{m}'u_1v_1^q=\delta_{m}u_2^qv_2+\delta_{m}'u_2v_2^q,\\
\gamma^{2pl}~~&:u_1^qv_1+u_1v_1^q=u_2^qv_2+u_2v_2^q.
\end{flalign*}
Similar to the proof of Subcase 1.4, we obtain that $\frac{v_2}{v_1}=\frac{u_1}{u_2}=(\frac{v_2}{v_1})^q\in \mathbb{F}_q^{*}$.

Subcase 1.6: $l=l',p= p',\theta =\theta',\delta_{j} = \delta'_{j}$ $(\forall j\in \{1,2,...,pl\})$ and $ \delta_{m} \ne \delta'_{m}$ $( pl < m\le r)$, where $j$ and $m$ are positive integers.

In this case, we compare the coefficients of $\gamma^0, \gamma^{2pl}$, $\gamma^{l+m}$ and the proof is similar to Subcase 1.5.

Case 2: $X\in \{U_{\delta_{1},\delta_2,...,\delta_{r},\theta,l}^{p}\} $ and $Y\in \{V_{\delta_{1}',\delta_2',...,\delta_{r}',l'}\}$.

Let $\bar{u}_1,\bar{u}_2 \in X,~\bar{v}_1,\bar{v}_2 \in Y$ be four nonzero elements and
\begin{flalign*}
\bar{u}_1&=u_1+\sum_{a=1}^{p}(\theta u_1^q+u_1)\delta_{al}\gamma^{al}+\sum_{b=1,b\notin \{ l,2l,...,pl\}}^{r}u_1\delta_{b}\gamma^b,\\
\bar{u}_2&=u_2+\sum_{a=1}^{p}(\theta u_2^q+u_2)\delta_{al}\gamma^{al}+\sum_{b=1,b\notin \{ l,2l,...,pl\}}^{r}u_2\delta_{b}\gamma^b ,\\
\bar{v}_1&=v_1+v_1^q\gamma^{l'}+\sum_{b=1,b\ne l'}^{r}v_1\delta_{b}'\gamma^b,\\
\bar{v}_2&=v_2+v_2^q\gamma^{l'}+\sum_{b=1,b\ne l'}^{r}v_1\delta_{b}'\gamma^b
\end{flalign*}
such that $\bar{u}_1\bar{v}_1=\bar{u}_2\bar{v}_2$.

If $l\ne l'$, we compare the coefficients of $\gamma^0,~\gamma^{\mathrm{min}(l,l')}$ and the proof is similar to Subcase 1.1.

If $l=l'$ and $p>1$, we compare the coefficients of $\gamma^0,~ \gamma^{l},~\gamma^{pl+r}$ and the proof is similar to Subcase 1.2.

If $l=l'$ and $p=1$, we compare the coefficients of $\gamma^0,~\gamma^{2l}$ and the proof is similar to Subcase 1.1.

Case 3: $X,~Y \in \{V_{\delta_{1},\delta_2,...,\delta_{r},l}\} $.

Let $X= V_{\delta_{1},\delta_2,...,\delta_{r},l} $, $Y=V_{\delta_{1}',\delta_2',...,\delta_{r}',l'}$ and $(\delta_{1},\delta_2,...,\delta_{r},l) \ne (\delta_{1}',\delta_2',...,\delta_{r}',l')$.
Suppose that $\bar{u}_1,\bar{u}_2 \in X, ~\bar{v}_1,\bar{v}_2 \in Y$ are four nonzero elements and
\begin{flalign*}
\bar{u}_1&=u_1+u_1^q\gamma^{l}+\sum_{b=1,b\ne l}^{r}u_1\delta_{b}\gamma^b,~\bar{u}_2=u_2+u_2^q\gamma^{l}+\sum_{b=1,b\ne l}^{r}u_1\delta_{b}\gamma^b ,\\
\bar{v}_1&=v_1+v_1^q\gamma^{l'}+\sum_{b=1,b\ne l'}^{r}v_1\delta_{b}'\gamma^b,~\bar{v}_2=v_2+v_2^q\gamma^{l'}+\sum_{b=1,b\ne l'}^{r}v_1\delta_{b}'\gamma^b
\end{flalign*}
such that $\bar{u}_1\bar{v}_1=\bar{u}_2\bar{v}_2$.

If $l\ne l'$, we compare the coefficients of $\gamma^0,~\gamma^{\mathrm{min}(l,l')}$ and the proof is similar to Subcase 1.1.

If $l=l'$ and $\delta_m\ne \delta_m$, where $1\le m\le r$ and $m\ne l$, we compare the coefficients of $\gamma^0,~\gamma^{l},~\gamma^{l+m}$ and the proof is similar to Subcase 1.3.

In all cases that $X,Y \in \{ U_{\delta_{1},\delta_2,...,\delta_{r},\theta,l}^p\};$ $X\in \{ U_{\delta_{1},\delta_2,...,\delta_{r},\theta,l}^p\}, Y\in \{ V_{\delta_{1},\delta_2,...,\delta_{r},l}\}$; $X,Y \in \{ V_{\delta_{1},\delta_2,...,\delta_{r},l}\}$, let $\bar{u}_1,\bar{u}_2 \in X, \bar{v}_1,\bar{v}_2 \in Y$ such that $\bar{u}_1\bar{v}_1=\bar{u}_2\bar{v}_2$, then we can obtain that $\frac{u_1}{u_2}=\frac{v_2}{v_1}\in \mathbb{F}_q^{*}$, which means that $\bar{u}_1\mathbb{F}_q=\bar{u}_2\mathbb{F}_q$ and $\bar{v}_1\mathbb{F}_q=\bar{v}_2\mathbb{F}_q$. From Lemma 2.4, we conclude that any two distinct subspaces $X$ and $Y$ satisfy the condition dim$(X \cap \alpha Y) \le 1$ for any $\alpha \in \mathbb{F}_{q^n}^{*}$.
\end{proof}
Combining the results in Lemmas 3.1 and 3.2, we construct a large cyclic $(n,2k-2,k)_q$-CDC via the union of orbits of Sidon spaces provided by (3.1) and (3.2) in the following theorem.
\begin{theo}
For any positive integer $k\ge 2$ and $n= (2r+1)k$ with $r \ge 2$, suppose that $U_{\delta_{1},\delta_2,...,\delta_{r},\theta,l}^{p}$ and $V_{\delta_{1},\delta_2,...,\delta_{r},l}$ are the subspaces constructed in (3.1) and (3.2), respectively. Assume that \begin{flalign*}
&\mathcal{C}_{\delta_{1},\delta_2,...,\delta_{r},\theta,l}^{p}=\{\alpha U_{\delta_{1},\delta_2,...,\delta_{r},\theta,l}^{p}: \alpha \in \mathbb{F}_{q^n}^{*} \},~\mathcal{D}_{\delta_{1},\delta_2,...,\delta_{r},l}=\{\alpha V_{\delta_{1},\delta_2,...,\delta_{r},l}: \alpha \in \mathbb{F}_{q^n}^{*} \}.
\end{flalign*}
Let
\begin{flalign*}
&\mathcal{D}=\bigcup_{l=1}^{r}\bigcup_{\delta_1,\delta_2,...,\delta_r}\mathcal{D}_{\delta_{1},\delta_2,...,\delta_{r},l},~\mathcal{C}^{1}=\bigcup_{l=1}^{r}\bigcup_{\delta_1,\delta_2,...,\delta_r} \bigcup_{\theta \in H} \mathcal{C}_{\delta_{1},\delta_2,...,\delta_{r},\theta,l}^{1},
\\ &\mathcal{C}^{p}=\bigcup_{l=\lfloor \frac{r}{p+1} \rfloor+1}^{\lfloor \frac{r}{p}\rfloor}\bigcup_{\delta_1,\delta_2,...,\delta_r} \bigcup_{\theta \in H} \mathcal{C}_{\delta_{1},\delta_2,...,\delta_{r},\theta,l}^{p},
\end{flalign*}
where $p\in \{2,3,...,p_0\}$, then $\mathcal{C}=\mathcal{C}^1\cup \left(\bigcup\limits_{p=2}^{p_0}\mathcal{C}^{p}\right) \cup \mathcal{D}$ is a cyclic $(n,2k-2,k)_q$-CDC of size
\begin{flalign*}
\frac{\left((r+\sum_{i=2}^{p_0}(\lfloor \frac{r}{i}\rfloor-\lfloor \frac{r}{i+1} \rfloor))(q^k-1)(q-1)+r\right)(q^k-1)^{r-1}(q^n-1)}{q-1}.
\end{flalign*}
\end{theo}
\begin{proof}
By Lemmas 2.3 and 3.1, we have that $\mathcal{C}_{\delta_{1},\delta_2,...,\delta_{r},\theta,l}^{p}$ and $\mathcal{D}_{\delta_{1},\delta_2,...,\delta_{r},l}$ are cyclic $(n,2k-2,k)_q$-CDCs with cardinality $ \frac{q^n-1}{q-1}$. From Lemmas 2.4 and 3.2, we conclude that dim$(X \cap \alpha Y) \le 1$ for any two distinct subspaces $X,Y\in \{ U_{\delta_{1},\delta_2,...,\delta_{r},\theta,l}^{p}\} \cup \{V_{\delta_{1},\delta_2,...,\delta_{r},l} \}$ and any $\alpha \in \mathbb{F}_{q^n}^{*}$, which implies that
the minimum distance of $\mathcal{C}$ is still $2k-2$. Therefore, $\mathcal{C}$ is a cyclic $(n,2k-2,k)_q$-CDC with cardinality
$\frac{\left((r+\sum\limits_{i=2}^{p_0}(\lfloor \frac{r}{i}\rfloor-\lfloor \frac{r}{i+1} \rfloor))(q^k-1)(q-1)+r\right)(q^k-1)^{r-1}(q^n-1)}{q-1}$.
\end{proof}
\begin{rem}
Let $n=5k$, we have $p_0=2$ and
\begin{align*}
&U_{\delta_1,\delta_{2},\theta,1}^2=\left\{u+(\theta u^q+ u )\delta_{1}\gamma +(\theta u^q+ u )\delta_{2}\gamma^2 : u\in \mathbb{F}_{q^k} \right\},\\
&U_{\delta_{1},\delta_{2},\theta,l}^1=\left\{ u+ (\theta u^q+ u )\delta_{l}\gamma^{l}+u\delta_{3-l}\gamma^{3-l}: u\in \mathbb{F}_{q^k}\right\},\\
&V_{\delta_{1},\delta_{2},l}~~=\left\{ v+ v^q \gamma^{l}+v\delta_{3-l} \gamma^{3-l}: v\in \mathbb{F}_{q^k}\right\},
\end{align*}
based on Theorem 3.3. Therefore we give a cyclic $(n,2k-2,k)_q$-CDC of size $3(q^k-1)^2(q^n-1)$$+\frac{2(q^k-1)(q^n-1)}{q-1}$, while the construction in \cite{Li2024} 
yields a cyclic $(n,2k-2,k)_q$-CDC of size $ (q^k-1)(3q^k-2)\frac{q^n-1}{q-1}$ and the construction in \cite{Yu2024} 
yields a cyclic $(n,2k-2,k)_q$-CDC of size $ 2((q^k-1)^2(q^n-1)+\frac{(q^k-1)(q^n-1)}{q-1})$. It follows that the size of our cyclic $(n,2k-2,k)_q$-CDC is larger than the results in \cite{Li2024} and \cite{Yu2024}, 
cf. Table 1.
\end{rem}

\begin{example}
Let $q=3,k=3$ and $n=5\times 3=15.$ The construction in \cite{Li2024} 
yields a cyclic $(15,4,3)_3$-CDC of size
\begin{flalign*}
&(q^k-1)(3q^k-2)\frac{q^n-1}{q-1}=(3^3-1)(3\times 3^3-2)\times \frac{3^{15}-1}{2}\\
&=(3\times 3^3-2)\frac{(3^3-1)(3^{15}-1)}{2}
\end{flalign*}
and rate approximately $0.474$, according to \cite[Definition 2]{Koetter2008}. 
The construction in \cite{Yu2024} 
yields a cyclic $(15,4,3)_3$-CDC of size
\begin{flalign*}
&2((q^k-1)^2(q^n-1)+\frac{(q^k-1)(q^n-1)}{q-1})=2((3^3-1)^2(3^{15}-1)+\frac{(3^3-1)(3^{15}-1)}{2})\\
&=(4(3^3-1)+2) \frac{(3^3-1)(3^{15}-1)}{2}>(3\times 3^3-2)\frac{(3^3-1)(3^{15}-1)}{2}
\end{flalign*}
and rate approximately $0.480$.
By Theorem 3.3, we give a cyclic $(15,4,3)_3$-CDC of size
\begin{flalign*}
&3(q^k-1)^2(q^n-1)+\frac{2(q^k-1)(q^n-1)}{q-1}=3 (3^3-1)^2(3^{15}-1)+\frac{2(3^3-1)(3^{15}-1)}{2}\\
&=(6(3^3-1)+2) \frac{(3^3-1)(3^{15}-1)}{2} >(4(3^3-1)+2) \frac{(3^3-1)(3^{15}-1)}{2}
\end{flalign*}
and rate approximately $0.488$. Therefore, in this case, the size of cyclic $(15,4,3)_3$-CDC constructed in Theorem 3.3 is larger than those presented in \cite{Li2024,Yu2024} 
and is also the best known code.
\end{example}

\subsection{Constructions of cyclic CDCs from irreducible polynomials and a variant of Sidon spaces}
In this subsection, we use irreducible polynomials and a variant of Sidon spaces in \cite{Yu2024} 
 to present a construction of cyclic $(n,2k-2,k)_q$-CDCs.

For any positive integer $k \ge 2$ and $n=2rk$ with $r\ge 2$, suppose that $l$ and $p$ are positive integers, $\xi$ is a primitive element of $\mathbb{F}_{q^k}$. 
If $r=2$, $p_0=1$, otherwise, $p_0=\max\{ i\in \mathbb{N}^+: \lceil\frac{r}{i}\rceil-1>\lfloor \frac{r}{i+1} \rfloor  \}$. 
Assume that $f(x)$ is an irreducible polynomial over $\mathbb{F}_{q^k}$ with degree $2r$, $\gamma$ is a root of $f(x)$ and $ A \subset \left \{1, \xi, . . . , \xi^{q^k-2} \right \} $ satisfies $ f(0)\alpha \beta \ne 1$ for any $\alpha, \beta \in A$. 
Since $\alpha=\xi^i,\beta=\xi^j$ for $i,j\in \{0,1,...,q^k-2\}$, it follows that $ f(0)\alpha \beta \ne 1$ is equivalent to $f(0)\xi^{i+j}\ne 1$. From Lemma 2.5, we have $|A| = \left \lfloor \frac{q^k-2}{2} \right \rfloor $. 
Let
\begin{flalign}
&U_{\delta_{1},\delta_2,...,\delta_{r},\theta,l}^{p}=\left\{u+\sum_{a=1}^{p}(\theta u^q+u)\delta_{al}\gamma^{al}+\sum_{b=1,b\notin \{l,2l,...,pl\}}^{r}u\delta_{b}\gamma^b: u\in \mathbb{F}_{q^k} \right\},
\end{flalign}
where $\delta_1,\delta_2,...,\delta_{r-1} \in \mathbb{F}_{q^k}^{*},~\delta_r \in A,~ 1\le a,b\le r,~ \theta \in \{1,\xi,...,\xi^{q-2}\}$ and $1\le p \le p_0$. Moreover, if $p=1,$ then $1\le l \le r-1 $, otherwise, $\frac{r}{p+1} < l < \frac{r}{p}$. Let
\begin{flalign}
&V_{\delta_{1},\delta_2,...,\delta_{r},l}~~=\left\{v+v^q \gamma^{l}
+\sum_{b=1,b\ne l}^{r}v\delta_{b}\gamma^b: v\in \mathbb{F}_{q^k}\right\},
\end{flalign}
where $\delta_1,\delta_2,...,\delta_{r-1} \in \mathbb{F}_{q^k}^{*},~\delta_r \in A, ~1\le b \le r$ and $1\le l \le r-1 $.

\begin{lem}
Let $U_{\delta_{1},\delta_2,...,\delta_{r},\theta,l}^{p}$ be the subspace constructed in (3.3), where $\delta_1,\delta_2,...,\delta_{r-1} \in \mathbb{F}_{q^k}^{*},~\delta_r \in A,~ \theta \in \{1,\xi,...,\xi^{q-2}\}$ and $1\le p \le p_0$. Moreover, if $p=1,$ then $1\le l \le r-1 $, otherwise, $\frac{r}{p+1} < l < \frac{r}{p}$. Let $V_{\delta_{1},\delta_2,...,\delta_{r},l}$ represent the subspace constructed in (3.4), where $\delta_1,\delta_2,...,\delta_{r-1} \in \mathbb{F}_{q^k}^{*},~\delta_r \in A$ and $1\le l \le r-1.$ Then $U_{\delta_{1},\delta_2,...,\delta_{r},\theta,l}^{p}$ and $V_{\delta_{1},\delta_2,...,\delta_{r},l}$ are both Sidon spaces. Any two distinct subspaces $X$ and $Y$ in the set of $\{ U_{\delta_{1},\delta_2,...,\delta_{r},\theta,l}^{p}\} \cup \{ V_{\delta_{1},\delta_2,...,\delta_{r},l} \}$ satisfy the condition $dim(X \cap \alpha Y) \le 1$ for any $\alpha \in \mathbb{F}_{q^n}^{*}$. 
\end{lem}
\begin{proof}
Firstly, we prove that $U_{\delta_{1},\delta_2,...,\delta_{r},\theta,l}^{p}$ is a Sidon space. Let $\bar{u}_1,\bar{u}_2,\bar{v}_1,\bar{v}_2 \in U_{\delta_{1},\delta_{2},...,\delta_{r},\theta,l}^{p} $ be four nonzero elements and
\begin{flalign*}
\bar{u}_1&=u_1+\sum_{a=1}^{p}(\theta u_1^q+u_1)\delta_{al}\gamma^{al}+\sum_{b=1,b\notin \{ l,2l,...,pl\}}^{r}u_1\delta_{b}\gamma^b,\\
\bar{u}_2&=u_2+\sum_{a=1}^{p}(\theta u_2^q+u_2)\delta_{al}\gamma^{al}+\sum_{b=1,b\notin \{ l,2l,...,pl\}}^{r}u_2\delta_{b}\gamma^b,\\
\bar{v}_1&=v_1+\sum_{a=1}^{p}(\theta v_1^q+v_1)\delta_{al}\gamma^{al}+\sum_{b=1,b\notin \{ l,2l,...,pl\}}^{r}v_1\delta_{b}\gamma^b,\\
\bar{v}_2&=v_2+\sum_{a=1}^{p}(\theta v_2^q+v_2)\delta_{al}\gamma^{al}+\sum_{b=1,b\notin \{ l,2l,...,pl\}}^{r}v_2\delta_{b}\gamma^b
\end{flalign*}
such that $\bar{u}_1\bar{v}_1=\bar{u}_2\bar{v}_2$. Since $1,\gamma,...,\gamma^{2r-1} $ are linearly independent over $\mathbb{F}_{q^k}$, $f(x)$ is an irreducible polynomial over $\mathbb{F}_{q^k}$ with degree $2r$, the constant term $c =f(0)\ne 0$ and $\gamma $ is a root of $f(x)$, then we have
\begin{flalign*}
\gamma^0~&:u_1v_1(1-c\delta_{r}^2)=u_2v_2(1-c\delta_{r}^2).
\end{flalign*}
According to the property of the set $A$, we have $(1-c\delta_{r}^2)\ne 0$, which implies that $u_1v_1 =u_2v_2$. 
Substituting $u_1v_1 =u_2v_2$ into the coefficient equations of both sides of the identity $\bar{u}_1\bar{v}_1=\bar{u}_2\bar{v}_2$ with respect to $\gamma^l$, we obtain
\begin{flalign*}
\gamma^l~&:u_1^qv_1+u_1v_1^q=u_2^qv_2+u_2v_2^q.
\end{flalign*}
Set $\frac{u_1}{u_2}=\frac{v_2}{v_1}=\lambda \ne 0$, then $\lambda^q u_2^qv_1+\lambda u_2v_1^q=\lambda u_2^qv_1+\lambda^q u_2v_1^q$. It follows that $(\lambda-\lambda^q)u_2v_1^q=(\lambda-\lambda^q)u_2^qv_1$. If $\lambda-\lambda^q=0$, then $\frac{u_1}{u_2}=\frac{v_2}{v_1}=\lambda \in \mathbb{F}_q^{*}$. If $\lambda-\lambda^q \ne 0$, then $\frac{u_2}{v_1}=(\frac{u_2}{v_1})^q \in \mathbb{F}_q^{*}$.
By Definition 1, $U_{\delta_{1},\delta_2,...,\delta_{r},\theta,l}^{p}$ is a Sidon space. Similarly, $V_{\delta_{1},\delta_2,...,\delta_{r},l}$ is also a Sidon space.

Since $\gamma^{2r}$ is a linear combination of $1,\gamma,...,\gamma^{2r-1}$, for the cases that $X,Y \in \{ U_{\delta_{1},\delta_2,...,\delta_{r},\theta,l}^p\};$ $X\in \{ U_{\delta_{1},\delta_2,...,\delta_{r},\theta,l}^p\}, Y\in \{ V_{\delta_{1},\delta_2,...,\delta_{r},l}\}$; $X,Y \in \{ V_{\delta_{1},\delta_2,...,\delta_{r},l}\}$, the proof of any two distinct subspaces $X$ and $Y$ satisfying the condition $\mathrm{dim}(X \cap \alpha Y) \le 1$ for any $\alpha \in \mathbb{F}_{q^n}^{*}$, is similar to Lemma 3.2.
\end{proof}
Combining the result in Lemma 3.6, we construct a large cyclic $(n,2k-2,k)$-CDC via the union of Sidon spaces provided by (3.3) and (3.4) in the following theorem.
\begin{theo}
For any positive integer $k\ge 2$ and $n= 2rk$ with $r\ge 2$, suppose that $U_{\delta_{1},\delta_{2},...,\delta_{r},\theta,l}^{p}$ and $V_{\delta_{1},\delta_{2},...,\delta_{r},l}$ are the subspaces constructed in (3.3) and (3.4), respectively. Assume that \begin{flalign*}
&\mathcal{C}_{\delta_{1},\delta_2,...,\delta_{r},\theta,l}^{p}=\{\alpha U_{\delta_{1},\delta_2,...,\delta_{r},\theta,l}^{p}: \alpha \in \mathbb{F}_{q^n}^{*} \},~\mathcal{D}_{\delta_{1},\delta_2,...,\delta_{r},l}=\{\alpha V_{\delta_{1},\delta_2,...,\delta_{r},l}: \alpha \in \mathbb{F}_{q^n}^{*} \}.
\end{flalign*}Let
\begin{flalign*}
&\mathcal{D}=\bigcup_{l=1}^{r-1}\bigcup_{\delta_1,\delta_2,...,\delta_r}\mathcal{D}_{\delta_{1},\delta_2,...,\delta_{r},l},~\mathcal{C}^{1}=\bigcup_{l=1}^{r-1}\bigcup_{\delta_1,\delta_2,...,\delta_r} \bigcup_{\theta \in H} \mathcal{C}_{\delta_{1},\delta_2,...,\delta_{r},\theta,l}^{1},\\
&\mathcal{C}^{p}=\bigcup_{l=\lfloor \frac{r}{p+1} \rfloor+1}^{\lceil\frac{r}{p}\rceil-1}\bigcup_{\delta_1,\delta_2,...,\delta_r} \bigcup_{\theta \in H} \mathcal{C}_{\delta_{1},\delta_2,...,\delta_{r},\theta,l}^{p},
\end{flalign*}
where $p\in \{2,3,...,p_0\}$, then $\mathcal{C}=\mathcal{C}^1 \cup \left(\bigcup\limits_{p=2}^{p_0}\mathcal{C}^{p}\right) \cup \mathcal{D}$ is a cyclic $(n,2k-2,k)_q$-CDC of size
\begin{flalign*}
\frac{\left((r-1+\sum\limits_{i=2}^{p_0}(\lceil \frac{r}{i}\rceil-\lfloor \frac{r}{i+1} \rfloor-1))(q^k-1)(q-1)+r-1\right)(q^k-1)^{r-2}\lfloor \frac{q^k-2}{2}\rfloor(q^n-1)}{q-1}.
\end{flalign*}
\end{theo}
\begin{proof}
Similar to the proof of Theorem 3.3, we obtain that $\mathcal{C}$ is a cyclic $(n,2k-2,k)_q$-CDC of size $\frac{\left((r-1+\sum\limits_{i=2}^{p_0}(\lceil\frac{r}{i}\rceil-\lfloor \frac{r}{i+1} \rfloor-1))(q^k-1)(q-1)+r-1\right)(q^k-1)^{r-2}\lfloor \frac{q^k-2}{2}\rfloor(q^n-1)}{q-1}$.
\end{proof}

\begin{rem}
The cyclic $(n,2k-2,k)_q$-CDC provided by Theorem 3.3 in \cite{Yu2024} 
 is of size $\lfloor \frac{q^k-2}{2}\rfloor(r-1)(q^k-1)^{r-1}(q^n-1)$, where $k\ge 2,n=2rk$ with $r\ge 2$. Therefore, by Theorem 3.7, the size of our cyclic $(n,2k-2,k)$-CDC is larger than the result in \cite{Yu2024}, 
cf. Table 1.
\end{rem}


\begin{example}
Let $q=5,k=3$ and $n=16\times 3=48$. The cyclic $(48,4,3)_5$-CDC provided by Theorem 3.3 in \cite{Yu2024} 
is of size
$\lfloor \frac{q^k-2}{2}\rfloor(r-1)(q^k-1)^{r-1}(q^n-1)=\lfloor \frac{5^3-2}{2}\rfloor7(5^3-1)^{7}(5^{48}-1)$ and rate approximately $0.505$ according to \cite[Definition 2]{Koetter2008}. 
Since $r=8$, $p_0=2.$
By Theorem 3.7, we give a cyclic $(48,4,3)_5$-CDC of size
$\frac{(32(5^3-1)+7)(5^3-1)^6\lfloor \frac{5^3-2}{2}\rfloor(5^{48}-1)}{4}$ $>8\lfloor \frac{5^3-2}{2}\rfloor(5^3-1)^7(5^{48}-1)>\lfloor \frac{5^3-2}{2}\rfloor7(5^3-1)^{7}(5^{48}-1)$ and rate approximately $0.506$.
Therefore, in this case, the size of cyclic $(48,4,3)_5$-CDC constructed in Theorem 3.7 is larger than that presented in \cite{Yu2024} 
 and is also the best known code.
\end{example}


\begin{rem} 
In the case of $n=4k$, Theorem 3.7 produces a cyclic $(4k,2k-2,k)_q$-CDC of size $(q^k-1)\lfloor \frac{q^k-2}{2}\rfloor(q^{4k}-1)+ \lfloor \frac{q^k-2}{2}\rfloor \frac{q^{4k}-1}{q-1}$. By Lemma 2.6, both the Sphere-packing bound and Johnson bound on the size of the cyclic $(4k,2k-2,k)_q$-CDC are given by $\frac{(q^{4k}-1)(q^{4k-1}-1)}{(q^k-1)(q^{k-1}-1)}$. While these bounds are generally distinct, they coincide in the case of $n=4k$. Therefore, when $k$ tends to infinity, the ratio between the size of our CDC in Theorem 3.7 and the bounds in Lemma 2.6 is approximately equal to $\frac{1}{2}$.
 \end{rem}

\subsection{Constructions of cyclic CDCs from Subspace polynomials}
In this subsection, we generalize the constructions of cyclic CDCs by subspace polynomials in \cite[Theorem 4.2]{Zhao2019}, 
 which can be seen as the special case of the following theorem.

\begin{theo}
Let $e,n,k,s$ be positive integers and $k>2$, $1 \le s < k-1$. For $1\le i\le e $, suppose that $V_i$ is the set of roots of the subspace polynomial
\begin{flalign*}
P_{ V_{i}}(x)=x^{q^k}+\gamma_{s+1,i}x^{q^{s+1}}+\gamma_{s,i}x^{q^s}+\gamma_{s-1,i}x^{q^{s-1}}+\cdots +\gamma_{0,i}x \in \mathbb{F}_{q^n}[x]
\end{flalign*}
satisfying $\gamma_{s+1,i}\ne 0$, $\gamma_{s,i}\ne 0$ and $\gamma_{0,i}\ne 0$.
Let $V_i~(1\le i\le e)$ be contained in $\mathbb{F}_{q^N}$. For $1 \le i,j \le e$ and $\alpha \in \mathbb{F}_{q^N}^*$, denote the matrix $M_{\alpha}^{i,j}$ by

\small
$$
\begin{pmatrix}
r_{s+1,\alpha}^{q^{k-s-1}} & 0 & 0 & 0 & 0 & \cdots & 0 & 0 & 0 &1\\
r_{s,\alpha}^{q^{k-s-1}} & r_{s+1,\alpha}^{q^{k-s-2}} & 0 & 0 & 0 & \cdots & 0 & 0 & 0 &0 \\
r_{s-1,\alpha}^{q^{k-s-1}} & r_{s,\alpha}^{q^{k-s-2}} & r_{s+1,\alpha}^{q^{k-s-3}} & 0 & 0 & \cdots & 0 & 0 & 0 &0 \\
r_{s-2,\alpha}^{q^{k-s-1}} & r_{s-1,\alpha}^{q^{k-s-2}} & r_{s,\alpha}^{q^{k-s-3}} & r_{s+1,\alpha}^{q^{k-s-4}} & 0 & \cdots & 0 & 0 & 0 &0 \\
\vdots & \vdots & \vdots & \vdots & \vdots & \vdots & \vdots & \vdots & \vdots & \vdots \\
r_{0,\alpha}^{q^{k-s-1}} & r_{1,\alpha}^{q^{k-s-2}} & r_{2,\alpha}^{q^{k-s-3}} & r_{3,\alpha}^{q^{k-s-4}} & r_{4,\alpha}^{q^{k-s-5}} & \cdots &0 & 0 & 0 &0 \\
0 & r_{0,\alpha}^{q^{k-s-2}} & r_{1,\alpha}^{q^{k-s-3}} & r_{2,\alpha}^{q^{k-s-4}} & r_{3,\alpha}^{q^{k-s-5}} & \cdots &0 & 0 & 0 &0 \\
\vdots & \vdots & \vdots & \vdots & \vdots & \vdots & \vdots & \vdots & \vdots & \vdots \\
0 & 0 & 0 & 0 & 0 & \cdots & r_{s-1,\alpha}^{q^2} & r_{s,\alpha}^q & r_{s+1,\alpha} &\gamma_{s+1,i}\\
0 & 0 & 0 & 0 & 0 & \cdots & r_{s-2,\alpha}^{q^2} & r_{s-1,\alpha}^q & r_{s,\alpha} &\gamma_{s,i}\\
\vdots & \vdots & \vdots & \vdots & \vdots & \vdots & \vdots & \vdots & \vdots & \vdots \\
0 & 0 & 0 & 0 & 0 & \cdots & 0 & 0 & r_{0,\alpha} &\gamma_{0,i}
\end{pmatrix}_{(k+1)\times (k-s+1)},
$$
where $r_{t,\alpha}=\gamma_{t,i}-\gamma_{t,j}\alpha^{q^k-q^t}$ for $0 \le t \le s+1$. If the following two conditions are satisfied:

(1) for any $1 \le i, j \le e$ and any  $\alpha \in \mathbb{F}_{q^N}^* \setminus \left( \mathbb{F}_{q^{k - s - 1}}^* \cup \mathbb{F}_{q^k}^* \right)$, the rank of  $M_{\alpha}^{i,j}$ is  $k - s + 1$; 

(2) for any  $1 \le i, j \le e$ with $i \ne j$, there exists  
$h \in \{1, 2, \dots, s+1\}$ such that  $gcd(h,n)=1$ and 
\begin{flalign*}
\left( \frac{\gamma_{0,i}}{\gamma_{0,j}} \right)^{\frac{q^h - 1}{q - 1}} 
\ne \left( \frac{\gamma_{0,i}}{\gamma_{0,j}} \left( \frac{\gamma_{h,i}}{\gamma_{h,j}} \right)^{-1} \right)^{\frac{q^k - 1}{q - 1}},
\end{flalign*}
then $\mathcal{C}=\bigcup\limits_{i=1}^e\{\alpha V_i : \alpha \in \mathbb{F}_{q^N}^*\}$ is a cyclic $(N,d,k)_q$-CDC of size $e\frac{q^N-1}{q-1}$, where $d \ge 2k-2s$.
\end{theo}
\begin{proof}
Firstly, we prove that $\{\alpha V_i: \alpha \in \mathbb{F}_{q^N}^*\}$ is a cyclic $(N,d,k)_q$-CDC of size $\frac{q^N-1}{q-1}$, where $d \ge 2k-2s$ for $1\le i\le e$.
Note that $V_i$ is $\mathbb{F}_{q^t}$-linear if and only if $P_{V_i}(x)$ is a $q^t$-polynomial. Since gcd$(s,s+1)=1$, then $P_{V_i}(x)$ is a $q$-polynomial and $|\{\alpha V_i: \alpha \in \mathbb{F}_{q^N}^*\}|= \frac{q^N-1}{q-1}$. Therefore, it is enough to prove that $d(V_i , \alpha V_i)\ge 2k-2s, ~\forall \alpha \in \mathbb{F}_{q^N}^* \backslash \mathbb{F}_{q}^*$. According to Lemma 2.9,
\begin{flalign*}
P_{\alpha V_i}(x)=x^{q^k}+\alpha^{q^k-q^{s+1}} \gamma_{s+1,i}x^{q^{s+1}}+\alpha^{q^k-q^{s}}\gamma_{s,i}x^{q^s}+\alpha^{q^k-q^{s-1}}\gamma_{s-1,i}x^{q^{s-1}}+\cdots+\alpha^{q^k-1}\gamma_{0,i}x.
\end{flalign*}
Let
\begin{flalign*}
R_{\alpha}(x)=P_{V_i}(x)-P_{\alpha V_i}(x)=r_{s+1,\alpha}x^{q^{s+1}}+r_{s,\alpha}x^{q^s}+r_{s-1,\alpha}x^{q^{s-1}}+\cdots +r_{0,\alpha}x,
\end{flalign*}
where $r_{t,\alpha}=\gamma_{t,i}(1-\alpha^{q^k-q^t})$ for $0 \le t \le s+1$.
Since
\begin{flalign*}
\mathrm{dim}(V_i \cap \alpha V_i)=\mathrm{log}_q(\mathrm{deg}(\mathrm{gcd}(P_{V_i}(x),P_{\alpha V_i}(x))))=\mathrm{log}_q(\mathrm{deg}(\mathrm{gcd}(P_{V_i}(x),R_{\alpha}(x)))),
\end{flalign*}
it is sufficient to prove that deg(gcd$(P_{V_i}(x),R_{\alpha}(x)))\le q^s,~\forall \alpha \in \mathbb{F}_{q^N}^* \backslash \mathbb{F}_{q}^*$. We distinguish between three cases that $r_{0,\alpha}=0;~r_{s+1,\alpha}=0$; $r_{s+1,\alpha}\ne 0,r_{0,\alpha}\ne 0$.
We can obtain that deg(gcd$(P_{V_i}(x),R_{\alpha}(x)))\le q^s$, $\forall \alpha \in \mathbb{F}_{q^N}^* \backslash \mathbb{F}_{q}^*$ for above three cases, similar to the proof of \cite[Theorem 4.2]{Zhao2019}. 

Next we prove that
\begin{flalign*}
\mathrm{dim}(V_i \cap \alpha V_j)\le 2s ~\mathrm{for} ~\mathrm{any} ~\alpha \in \mathbb{F}_{q^N}^* ~\mathrm{and} ~1\le i\ne j \le e.
\end{flalign*}
Suppose that
\begin{flalign*}
R_{\alpha}(x)=P_{V_{i}}(x)-P_{\alpha V_j}(x)=r_{s+1,\alpha}x^{q^{s+1}}+r_{s,\alpha}x^{q^s}+r_{s-1,\alpha}x^{q^{s-1}}+\cdots+r_{0,\alpha}x,
\end{flalign*}
where $r_{t,\alpha}=\gamma_{t,i}-\gamma_{t,j}\alpha^{q^k-q^t}$ for $0 \le t \le s+1$. If $V_{i}=\alpha V_j$, then
\begin{flalign*}
\gamma_{h,i}-\gamma_{h,j}\alpha^{q^k-q^{h}}=0~ \mathrm{and} ~\gamma_{0,i}-\gamma_{0,j}\alpha^{q^k-1}=0,
\end{flalign*}
where $h\in \{1,2,...,s+1\}$.
It follows that $\frac{\gamma_{0,i}}{\gamma_{0,j}}=\alpha^{q^k-1}$ and $\frac{\gamma_{0,i}}{\gamma_{0,j}}\left(\frac{\gamma_{h,i}}{\gamma_{h,j}}\right)^{-1}=\alpha^{q^{h}-1}$, which means that
\begin{flalign*}
\left(\frac{\gamma_{0,i}}{\gamma_{0,j}}\right)^{\frac{q^{h}-1}{q-1}}=\alpha^{\frac{(q^{h}-1)(q^k-1)}{q-1}}
\mathrm{and}~ \left(\frac{\gamma_{0,i}}{\gamma_{0,j}}\left(\frac{\gamma_{h,i}}{\gamma_{h,j}}\right)^{-1}\right)^{\frac{q^k-1}{q-1}}=\alpha^{\frac{(q^{h}-1)(q^k-1)}{q-1}}.
\end{flalign*}
This leads to
\begin{flalign*}
\left(\frac{\gamma_{0,i}}{\gamma_{0,j}}\right)^{\frac{q^{h}-1}{q-1}}=\alpha^{\frac{(q^k-1)(q^{h}-1)}{q-1}} =\left(\frac{\gamma_{0,i}}{\gamma_{0,j}}\left(\frac{\gamma_{h,i}}{\gamma_{h,j}}\right)^{-1}\right)^{\frac{q^{k}-1}{q-1}},
\end{flalign*}
which contradicts our condition. Therefore, $r_{s+1,\alpha}$ and $r_{0,\alpha}$ cannot be zero simultaneously. At this point, taking arguments similar to those used in the proof of \cite[Theorem 4.2]{Zhao2019}, 
we obtain the desired result.
\end{proof}

In particular, we discuss the special case of $q=2$ in the following corollary. Moreover, we replace the coefficient constraints in Theorem 3.11 with more specific conditions, which lead to the following result.

\begin{coro}
Let $k,n,s,r$ be positive integers and $1 \le s < k-1$. For each $1 \le i \le e$, suppose $V_i$ is the root set of the subspace polynomial  
\begin{flalign*}
P_{V_i}(x) = x^{2^k} + \gamma_{s+1,i} x^{2^{s+1}} + \gamma_{s,i} x^{2^s} + \gamma_{s-1,i} x^{2^{s-1}} + \cdots + \gamma_{0,i} x \in \mathbb{F}_{2^n}[x],
\end{flalign*}
satisfying $\gamma_{s+1,i} \ne 0$, $\gamma_{s,i} \ne 0$ and $\gamma_{0,i} \ne 0$. Assume $V_i$ is contained in $\mathbb{F}_{q^N}$. If   the following two conditions are satisfied:

(1) for any  $1 \le i, j \le e$ and any $\alpha \in \mathbb{F}_{q^N}^* \setminus \left( \mathbb{F}_{q^{k - s - 1}}^* \cup \mathbb{F}_{q^k}^* \right)$,  the rank of  $M_{\alpha}^{i,j}$ is  $k - s + 1$; 

(2) for any $1\le i , j \le e$ with $i\ne j$, we have  $\gamma_{0,i} \ne \gamma_{0,j}$, and there exists $h \in \{1, 2, \dots, s+1\}$ such that  
$\mathrm{gcd}(h,n) = 1, \gamma_{h,i} = \gamma_{0,i} \in \mathbb{F}_{q^n}^*~ \text{and}~ \gamma_{h,j} = \gamma_{0,j} \in \mathbb{F}_{q^n}^*$,\\
then $\mathcal{C} = \bigcup_{i=1}^e \{ \alpha V_i : \alpha \in \mathbb{F}_{2^N}^* \}$ is a cyclic $(N,d,k)_2$-CDC of size $e(2^N-1)$, where $d \ge 2k-2s$.
\end{coro}

\begin{proof}
 For $q = 2$, there exists $h \in \{1, 2, \dots, s+1\}$ such that $\gamma_{h,i} = \gamma_{0,i} \in \mathbb{F}_{q^n}^*$ and $\gamma_{h,j} = \gamma_{0,j} \in \mathbb{F}_{q^n}^*$, Condition 2 of Theorem 3.11 becomes $\left(\frac{\gamma_{0,i}}{\gamma_{0,j}}\right)^{2^h - 1} \ne 1$.  
Note that $\mathrm{gcd}(h,n) = 1$. If $\left(\frac{\gamma_{0,i}}{\gamma_{0,j}}\right)^{2^h - 1} = 1$, then $\frac{\gamma_{0,i}}{\gamma_{0,j}} \in \mathbb{F}_{2^h}^* \cap \mathbb{F}_{2^n}^* = \mathbb{F}_2^*$, i.e., $\gamma_{0,i} = \gamma_{0,j}$, which contradicts the assumption. Therefore, Condition 2 of Theorem 3.11 is satisfied, and the rest of the proof follows similarly to Theorem 3.11.
\end{proof}

\begin{rem}
Zhao et al.~\cite{Zhao2019} introduced a construction of cyclic \( (N, 2k - 2, k)_q \)-CDCs based on subspace polynomials.  
We extend this work by presenting a more general construction of cyclic $ (N, d, k)_q $-CDCs of size $e \frac{q^N - 1}{q - 1} $, where $d \ge 2k - 2s $.  
Specifically, Theorem~3.11 includes Theorem~4.2 in~\cite{Zhao2019} in the case $s = 1$ and $e = 1 $.
Moreover, compared to previous constructions based on trinomials, our coefficient selection allows greater flexibility. Under certain parameter settings,  we can construct optimal cyclic CDCs with either more admissible values of $N$ or larger size (see Example 3.14).
\end{rem}

\begin{example}
Let $s = 1$, $k = 3$, $q = 2$ and $n = 2$. According to Corollary 3.13 of \cite{Chen2018} and Theorem 4.2 in \cite{Zhao2019}, under the case $N = 14$, the constructed optimal cyclic CDC has size at most $2^{14} - 1$.  
In contrast, Corollary 3.12 in this paper allows the construction of an optimal cyclic CDC of size $3(2^{14} - 1)$.  
Take $h = 1$, then we have
\begin{align*}
&P_{V_1}(x)=x^{2^3}+\xi x^{2^2}+x^{2}+x,~P_{V_2}(x)=x^{2^3}+\xi x^{2^2}+\xi x^{2}+\xi x, \\
&P_{V_3}(x)=x^{2^3}+\xi^2 x^{2^2}+\xi^2 x^{2}+\xi^2x.
\end{align*}
It is noted that the polynomial coefficients satisfy Condition 1 in Corollary 3.12.  
For any $1 \le i, j \le 3$ and any $\alpha \in \mathbb{F}_{2^{14}}^* \setminus \left( \mathbb{F}_2^* \cup \mathbb{F}_{2^3}^* \right)$, the following matrices $M^{i,j}_{\alpha}$ are defined:

\vspace{0.5em}

\noindent
\begin{minipage}[t]{0.48\textwidth}
\[
M^{1,1}_{\alpha} =
\begin{pmatrix}
\left( \xi(1 - \alpha^4) \right)^2 & 0 & 1 \\
\left( 1 - \alpha^6 \right)^2 & \xi(1 - \alpha^4) & \xi \\
\left( 1 - \alpha^7 \right)^2 & 1 - \alpha^6 & 1 \\
0 & 1 - \alpha^7 & 1
\end{pmatrix}
\]
\end{minipage}
\hfill
\begin{minipage}[t]{0.48\textwidth}
\[
M^{1,2}_{\alpha} =
\begin{pmatrix}
\left( \xi(1 - \alpha^4) \right)^2 & 0 & 1 \\
\left( 1 - \xi \cdot \alpha^6 \right)^2 & \xi(1 - \alpha^4) & \xi \\
\left( 1 - \xi \cdot \alpha^7 \right)^2 & 1 - \xi \cdot \alpha^6 & 1 \\
0 & 1 - \xi \cdot \alpha^7 & 1
\end{pmatrix}
\]
\end{minipage}

\vspace{1.5em}

\noindent
\begin{minipage}[t]{0.48\textwidth}
\[
M^{1,3}_{\alpha} =
\begin{pmatrix}
\left( \xi - \xi^2 \cdot \alpha^4 \right)^2 & 0 & 1 \\
\left( 1 - \xi^2 \cdot \alpha^6 \right)^2 & \xi - \xi^2 \cdot \alpha^4 & \xi \\
\left( 1 - \xi^2 \cdot \alpha^7 \right)^2 & 1 - \xi^2 \cdot \alpha^6 & 1 \\
0 & 1 - \xi^2 \cdot \alpha^7 & 1
\end{pmatrix},
\]
\end{minipage}
\hfill
\begin{minipage}[t]{0.48\textwidth}
\[
M^{2,2}_{\alpha} =
\begin{pmatrix}
\left( \xi(1 - \alpha^4) \right)^2 & 0 & 1 \\
\left( \xi(1 - \alpha^6) \right)^2 & \xi(1 - \alpha^4) & \xi \\
\left( \xi(1 - \alpha^7) \right)^2 & \xi(1 - \alpha^6) & \xi \\
0 & \xi(1 - \alpha^7) & \xi
\end{pmatrix},
\]
\end{minipage}

\vspace{1.5em}

\noindent
\begin{minipage}[t]{0.48\textwidth}
\[
M^{2,3}_{\alpha} =
\begin{pmatrix}
\left( \xi - \xi^2 \cdot \alpha^4 \right)^2 & 0 & 1 \\
\left( \xi - \xi^2 \cdot \alpha^6 \right)^2 & \xi - \xi^2 \cdot \alpha^4 & \xi \\
\left( \xi - \xi^2 \cdot \alpha^7 \right)^2 & \xi - \xi^2 \cdot \alpha^6 & \xi \\
0 & \xi - \xi^2 \cdot \alpha^7 & \xi
\end{pmatrix},
\]
\end{minipage}
\hfill
\begin{minipage}[t]{0.48\textwidth}
\[
M^{3,3}_{\alpha} =
\begin{pmatrix}
\left( \xi^2(1 - \alpha^4) \right)^2 & 0 & 1 \\
\left( \xi^2(1 - \alpha^6) \right)^2 & \xi^2(1 - \alpha^4) & \xi^2 \\
\left( \xi^2(1 - \alpha^7) \right)^2 & \xi^2(1 - \alpha^6) & \xi^2 \\
0 & \xi^2(1 - \alpha^7) & \xi^2
\end{pmatrix}.
\]
\end{minipage}
Using Magma, it can be verified that for any $1 \le i, j \le 3$ and any $\alpha \in \mathbb{F}_{2^{14}}^* \setminus \left( \mathbb{F}_2^* \cup \mathbb{F}_{2^3}^* \right)$, the rank of $M^{i,j}_{\alpha}$ is 3. Therefore, Condition 2 of Corollary 3.12 is satisfied. It follows that $\mathcal{C} = \bigcup_{i=1}^3 \{ \alpha V_i : \alpha \in \mathbb{F}_{2^{14}}^* \}$ is a cyclic $(14,4,3)_2$-CDC of size $3(2^{14}-1)$.
Moreover, under this parameter setting, our construction can provide optimal cyclic CDCs with $N=6$, which cannot be achieved by trinomials.
\end{example}

\section{Conclusions and future work}
In this paper, we improved the cardinality of some cyclic $(n,2k-2,k)_q$-CDCs by further exploring the ideas proposed in \cite{Li2024} and \cite{Yu2024}. 
More explicitly, we presented several new constructions of Sidon spaces and ensured that the union of optimal cyclic orbit codes generated by these Sidon spaces did not change the minimum distance. Consequently, in Theorems 3.3 and 3.7, we provided two new constructions of cyclic $(n,2k-2,k)_q$-CDCs in the cases where $n=(2r+1)k$ and $n=2rk$, respectively. Moreover, our cyclic $(n,2k-2,k)_q$-CDCs have larger sizes than the best known results. 
Table 1 provides a comparison of the sizes of the constructions. 
 In the case of $n=4k$, when $k$ goes to infinity, the ratio between the size of our cyclic $(4k,2k-2,k)_q$-CDC and the Sphere-packing bound (Johnson bound) is approximately equal to $\frac{1}{2}$. 
Furthermore, we provided a new construction of cyclic $(N,d,k)_q$-CDC by subspace polynomials, where $1\le s< k-1$ and $d\ge 2k-2s$. 
Our construction generalizes previous results and, under certain parameter settings, provides cyclic CDCs with larger sizes or more admissible values of $N$ than those based on trinomial.

Recent developments have further advanced the theory of rank-metric codes based on linearized polynomials.
Neri et al.~\cite{Neri2022} extend the family of rank-metric codes which are $\mathbb{F}_{q^{2t}}$-linear MRD codes of dimension 2 in the space of linearized polynomials over $\mathbb{F}_{q^{2t}}$, where $t$ is any integer greater than 2. 
Bartoli and Ghiandoni~\cite{Bartoli2025} present results on the classification of $\mathbb{F}_{q^n}$-linear MRD codes of dimension 3. Therefore, combining these perspectives, a potential research avenue is to extend the family \(\mathcal{C}_{h,t,\sigma}\) described in \cite{Neri2022} to higher dimensions.
Next, we will consider this as a primary direction for our future work.

\section*{Acknowledgement}
The authors thank the Associated Editor and anonymous reviewers for their valuable comments that improved the presentation and quality of this paper. G. Wang, M. Xu and Y. Gao are supported by the National Natural Science Foundation of China (No. 12301670), the Natural Science Foundation of Tianjin (No. 23JCQNJC00050), the Scientific Research Project of Tianjin Education Commission (No. 2023ZD041), the Fundamental Research Funds for the Central Universities of China (No. 3122024PT24) and the Graduate Student Research and Innovation Fund of Civil Aviation University of China (No. 2023YJSKC06006).





\begin{thebibliography}{99}

\bibitem{Ahlswede2000} 
Ahlswede R., Cai N., Li S., Yeung R.W.: Network information flow, IEEE Transactions on Information Theory, 46(4), 1204-1216 (2000).

\bibitem{Bachoc2017} 
Bachoc C., Serra O., Zémor G.: An analogue of Vosper's theorem for extension fields, Mathematical Proceedings of the Cambridge Philosophical Society, 163(3), 423-452 (2017).


\bibitem{Bartoli2025}
Bartoli D., Ghiandoni S.: Algebraic–geometric analysis of three-dimensional maximum rank distance codes, Designs, Codes and Cryptography, 93(3), 611-631 (2025).

\bibitem{Ben-Sasson2016} 
Ben-Sasson E., Etzion T., Gabizon A., Raviv N.: Subspace polynomials and cyclic subspace codes, IEEE Transactions on Information Theory, 62(3), 1157-1165 (2016).

\bibitem{Ben-Sasson2010} 
Ben-Sasson E., Kopparty S., Radhakrishnan J.: Subspace polynomials and limits to list decoding of Reed-Solomon codes, IEEE Transactions on Information Theory, 56(1), 113-120 (2010).



\bibitem{Braun2016} 
Braun M., Etzion T., Ostergard P., Vardy A., Wassermann A.: Existence of $q$-Analogs of Steiner Systems, Forum of Mathematics, Pi, 4(e7), 1-14 (2016).

\bibitem{Cai2002} 
Cai N., Yeung R.W.: Network coding and error correction, Proceedings of the IEEE Information Theory Workshop, 119-122 (2002).

\bibitem{Chen2018} 
Chen B., Liu H.: Constructions of cyclic constant dimension codes, Designs, Codes and Cryptography, 86(6), 1267-1279 (2018).

\bibitem{Etzion2011} 
Etzion T., Vardy A.: Error-correcting codes in projective space, IEEE Transactions on Information Theory, 57(2), 1165-1173 (2011).

\bibitem{Feng2021} 
Feng T., Wang Y.: New constructions of large cyclic subspace codes and Sidon spaces, Discrete Mathematics, 344(4), 112273 (2021).

\bibitem{Gluesing-Luerssen2021} 
Gluesing-Luerssen H., Hunter L.: Distance distributions of cyclic orbit codes, Designs, Codes and Cryptography, 89(3), 447-470 (2021).

\bibitem{Gluesing-Luerssen2015} 
Gluesing-Luerssen H., Morrison K., Troha C.: Cyclic orbit codes and stabilizer subfields, Advances in Mathematics of Communications, 9(2), 177-197 (2015).



\bibitem{Ho2006} 
Ho T., Médard M., Koetter R., Karger D.R., Effros M., Shi J., Leong B.: A random linear network coding approach to multicast, IEEE Transactions on Information Theory, 52(10), 4413-4430 (2006).

\bibitem{Koetter2008} 
Koetter R., Kschischang F.: Coding for errors and erasures in random network coding, IEEE Transactions on Information Theory, 54(8), 3579-3591 (2008).

\bibitem{Koetter2002} 
Koetter R., Médard M.: Beyond routing: an algebraic approach to network coding, Proceedings. Twenty-First Annual Joint Conference of the IEEE Computer and Communications Societies, 122-130 (2002).






\bibitem{Kohnert2008} 
Kohnert A., Kurz S.: Construction of large constant dimension codes with a prescribed minimum distance, Mathematical Methods in Computer Science, 5393, 31-42 (2008).

\bibitem{Lao2022} 
Lao H., Chen H.: New constant dimension subspace codes from multilevel linkage construction, Advances in Mathematics of Communications, 18(4), 956-966 (2022).

\bibitem{Lidl1997} 
Lidl R., Niederreiter H.: Finite Fields, Cambridge University Press, Cambridge, (1997).

\bibitem{Li2023} 
Li Y., Liu H.: Cyclic constant dimension subspace codes via the sum of Sidon spaces, Designs, Codes and Cryptography, 91(4), 1193-1207 (2023).

\bibitem{Li2024} 
Li Y., Liu H., Mesnager S.: New constructions of constant dimension subspace codes with large sizes. Designs, Codes and Cryptography, 92(5), 1423-1437 (2024).

\bibitem{Liu2023} 
Liu X., Shi T., Niu M., Shen L., Gao Y.: New constructions of Sidon spaces and cyclic subspace codes, IEICE on Fundamentals of Electronics Communications and Computer Sciences, E106-A(8), 1062-1066 (2023).

\bibitem{Niu2022} 
Niu M., Xiao J., Gao Y.: New constructions of large cyclic subspace codes via Sidon spaces, Advances in Mathematics of Communications, 18(4), 1123-1137 (2022).

\bibitem{Niu2020} 
Niu Y., Yue Q., Wu Y.: Several kinds of large cyclic subspace codes via Sidon spaces, Discrete Mathematics, 343(5), 111788 (2020).

\bibitem{Neri2022}
Neri A., Santonastaso P., Zullo F.: Extending two families of maximum rank distance codes, Finite Fields and Their Applications, 81, 102045 (2022).

\bibitem{Otal2017} 
Otal K., Özbudak F.: Cyclic subspace codes via subspace polynomials, Designs, Codes and Cryptography, 85(2), 191-204 (2017).

\bibitem{Roth2017} 
Roth R.M., Raviv N., Tamo I.: Construction of Sidon spaces with applications to coding, IEEE Transactions on Information Theory, 64(6), 4412-4422 (2017).

\bibitem{Trautmann2013} 
Trautmann A.L., Manganiello F., Braun M., Rosenthal J.: Cyclic orbit codes, IEEE Transactions on Information Theory, 59(11), 7386-7404 (2013).

\bibitem{Yu2024} 
Yu S., Ji L.: Two new constructions of cyclic subspace codes via Sidon spaces, Designs, Codes and Cryptography, 92(11), 3799-3811 (2024).

\bibitem{Zhang2021} 
Zhang H., Cao X.: Further constructions of cyclic subspace codes, Cryptography and Communications, 13(2), 245-262 (2021).

\bibitem{Zhang2022} 
Zhang H., Cao X.: Constructions of Sidon spaces and cyclic subspace codes, Frontiers of Mathematics in China, 17(2), 275-288 (2022).



\bibitem{Zhang2023c} 
Zhang H., Tang C.: Constructions of large cyclic constant dimension codes via Sidon spaces, Designs, Codes and Cryptography, 91(1), 29-44 (2023).

\bibitem{Zhang2023} 
Zhang H., Tang C.: Further constructions of large cyclic subspace codes via Sidon spaces, Linear Algebra and Its Applications, 661, 106-115 (2023).

\bibitem{Zhang2023b} 
Zhang H., Tang C., Hu X.: New constructions of Sidon spaces and large cyclic constant dimension codes, Computational and Applied Mathematics, 42(5), 230 (2023).

\bibitem{Zhang2022b} 
Zhang T., Ge G.: New constructions of Sidon spaces, Journal of Algebraic Combinatorics, 55(3), 781-794 (2022).

\bibitem{Zullo2023} 
Zullo F.: Multi-orbit cyclic subspace codes and linear sets, Finite Fields and Their Applications, 87, 102153 (2023).

\bibitem{Zhao2019} 
Zhao W., Tang X.: A characterization of cyclic subspace codes via subspace polynomials, Finite Fields and Their Applications, 57, 1-12 (2019).



















\end{thebibliography}
\end{document}